\documentclass[a4paper,11pt]{article}
\usepackage{lmodern}
\usepackage[utf8]{inputenc}
\usepackage[T1]{fontenc}
\usepackage[USenglish]{babel}
\usepackage{amsmath,amssymb,amsthm}

\usepackage{fullpage}

\usepackage{graphicx}
\usepackage{enumerate}
\usepackage[numbers,sort&compress]{natbib}
\usepackage[labelsep=period,labelfont={sc},textfont=it,margin=\parindent,font=small,format=plain]{caption}
\usepackage[labelfont=bf,font=small,textfont=rm,margin=\parindent,format=hang,labelformat=parens]{subcaption}
\usepackage[%
  pdftitle={Bounded, minimal, and short representations of unit interval and unit circular-arc graphs},%
  pdfauthor={Francisco J.\ Soulignac},%
  pdfcreator={},%
  pdfsubject={},%
  pdfkeywords={},%
  colorlinks=true,%
  linkcolor=blue,%
  citecolor=blue,%
  urlcolor=blue]{hyperref}  
\usepackage{doi}
\usepackage{xspace}
  
\usepackage{comment}

\newtheorem{theorem}{Theorem}
\newtheorem{lemma}[theorem]{Lemma}
\newtheorem{corollary}[theorem]{Corollary}
\newtheorem{observation}[theorem]{Observation}

\title{Bounded, minimal, and short representations of unit interval and unit circular-arc graphs}
\author{Francisco J.\ Soulignac\thanks{CONICET and Departamento de Ciencia y Tecnología, Universidad Nacional de Quilmes, Bernal, Argentina.}}

\date{\normalsize\texttt{francisco.soulignac@unq.edu.ar}}

\newcommand{\A}{\ensuremath{\mathcal{A}}}
\newcommand{\C}{\ensuremath{\mathcal{C}}}
\newcommand{\M}{\ensuremath{\mathcal{M}}}
\newcommand{\N}{\ensuremath{\mathcal{N}}}
\newcommand{\MP}{\ensuremath{\mathcal{P}}}
\newcommand{\U}{\ensuremath{\mathcal{U}}}
\newcommand{\W}{\ensuremath{\mathcal{W}}}
\newcommand{\T}{\ensuremath{\mathcal{T}}}

\newcommand{\SEP}{\ensuremath{\textrm{sep}}}
\newcommand{\MAG}{\ensuremath{\textrm{len}}}
\newcommand{\REM}{\ensuremath{\textrm{ext}}}
\newcommand{\CONST}{\ensuremath{\textrm{const}}}
\newcommand{\REACH}{\ensuremath{\textrm{reach}}}

\newcommand{\Noses}{\ensuremath{\nu}}
\newcommand{\Hollows}{\ensuremath{\eta}}
\newcommand{\Steps}{\ensuremath{\sigma}}
\newcommand{\Bounds}{\ensuremath{\beta}}
\newcommand{\Height}{\ensuremath{h}}
\newcommand{\Column}{\ensuremath{c}}
\newcommand{\Len}{\ensuremath{\ell}}
\newcommand{\Circ}{\ensuremath{c}}
\newcommand{\Bound}{\ensuremath{d}}
\newcommand{\LeftBound}{\ensuremath{\Bound_\ell}}
\newcommand{\RightBound}{\ensuremath{\Bound_r}}
\newcommand{\Thres}{\ensuremath{d}}
\newcommand{\BegDist}{\ensuremath{\Bound_s}}
\newcommand{\Unit}{\U}
\newcommand{\Descriptor}{\ensuremath{u}}
\newcommand{\Jump}{\ensuremath{\delta}}
\newcommand{\Ratio}{\ensuremath{r}}
\newcommand{\RATIO}{\ensuremath{R}}
\newcommand{\Extra}{\ensuremath{e}}

\newcommand{\IDist}[1]{\ensuremath{{\rm\bf d}#1}}
\newcommand{\Dist}[1]{\ensuremath{{\rm\bf d}^*#1}}
\newcommand{\Syn}{\ensuremath{\mathcal{S}}}
\newcommand{\BoundedSyn}{\ensuremath{\mathcal{B}}}
\newcommand{\Red}{\ensuremath{\mathcal{R}}}
\newcommand{\RedMitas}{\ensuremath{\mathcal{T}}}

\newcommand{\sg}{\ensuremath{\mathrm{sg}}}

\newcommand{\PROBLEMFORMAT}[1]{\textsc{#1}\xspace}
\newcommand{\Rep}{\PROBLEMFORMAT{Rep}}
\newcommand{\BoundRep}{\PROBLEMFORMAT{BoundRep}}
\newcommand{\IntBoundRep}{\PROBLEMFORMAT{IntBoundRep}}
\newcommand{\BoundRepBoth}{\PROBLEMFORMAT{(Int)BoundRep}}
\newcommand{\uRep}{\PROBLEMFORMAT{$\Descriptor$-Rep}}
\newcommand{\MinUIG}{\PROBLEMFORMAT{MinUIG}}
\newcommand{\MinUCA}{\PROBLEMFORMAT{MinUCA}}
\newcommand{\MinPP}{\PROBLEMFORMAT{MinP$_q^k$}}
\newcommand{\MinPC}{\PROBLEMFORMAT{MinC$_q^k$}}

\newcommand{\itemref}[3]{\ensuremath{\rm (\hyperref[def:#2-#3]{#1}_{\ref{def:#2-#3}})}}
\newcommand{\sepref}[1]{\itemref{\SEP}{sep}{#1}}

\newcommand{\figurepath}{}

\begin{document}
\maketitle

\begin{abstract}
  We consider the unrestricted, minimal, and bounded representation problems for unit interval (UIG) and unit circular-arc (UCA) graphs.  In the unrestricted version, a proper circular-arc (PCA) model $\M$ is given and the goal is to obtain an equivalent UCA model $\Unit$.  We show a linear time algorithm with negative certification that can also be implemented to run in logspace.
  In the bounded version, $\M$ is given together with some lower and upper bounds that the beginning points of $\Unit$ must satisfy.  We develop a linear space $O(n^2)$ time algorithm for this problem.  
  Finally, in the minimal version, the circumference of the circle and the length of the arcs in $\Unit$ must be simultaneously as minimum as possible.  We prove that every UCA graph admits such a minimal model, and give a polynomial time algorithm to find it.  We also consider the minimal representation problem for UIG graphs.  As a bad result, we show that the previous linear time algorithm fails to provide a minimal model for some input graphs.  We fix this algorithm but, unfortunately, it runs in linear space $O(n^2)$ time.  Finally, we apply the minimal representation algorithms so as to find the minimum powers of paths and cycles that contain a given UIG and UCA models, respectively.

  \vspace*{.2\baselineskip} {\bf Keywords:} unit circular-arc graphs, unit interval graphs, recognition problem, bounded representation problem, minimal model, powers of paths and cycles.  
\end{abstract}

\section{Introduction}

In this article we are concerned with some recognition and representation problems for unit interval and unit circular-arc graphs.  A \emph{proper circular-arc (PCA) model} is a pair $\M = (C, \A)$ where $C$ is a circle and $\A$ is a family of inclusion-free arcs of $C$ in which no pair of arcs in $\A$ cover $C$. If some point of $C$ is crossed by no arcs, then $\M$ is a \emph{proper interval (PIG)} model.  \emph{Unit circular-arc (UCA)} and \emph{unit interval (UIG)} models correspond to the PCA and PIG models in which all the arcs have the same length, respectively.  Every PCA model $\M$ is associated with a graph $G(\M)$ that contains a vertex for each of its arcs, where two vertices are adjacent if and only if their corresponding arcs have a nonempty intersection.  A graph $G$ is a \emph{proper circular-arc (PCA)}  graph when it is isomorphic to $G(\M)$ for some PCA model $\M$.  In such a case, $G$ is said to \emph{admit} the model $\M$, while $\M$ is said to \emph{represent} $G$.  \emph{Proper interval (PIG)}, \emph{unit circular-arc (UCA)}, and \emph{unit interval (UIG)} graphs are defined analogously.

The recognition problem is well solved for UIG graphs.  Indeed, Roberts' PIG=UIG Theorem states that every PIG graph admits a UIG model~\cite{Roberts1969}.  Hence, it suffices to determine if $G$ is a PIG graph, a task that can be accomplished in linear time (e.g.~\cite{HellHuangSJDM2004/05}) or logspace~\cite{KoblerKuhnertLaubnerVerbitskySJC2011}.  Moreover, there are \emph{certifying} algorithms that exhibit either a PIG model or a forbidden induced subgraph according to whether the input graph is PIG or not.

Knowing that $G$ is a UIG graph tells us nothing about its UIG models.  In this article we deal with the stronger \emph{(unrestricted) representation} (\Rep) problem in which a UIG model $\Unit$ equivalent to an input PIG model $\M$ is to be found.  By \emph{equivalent} we mean that the extremes of $\Unit$ must appear in the same order as in $\M$.  The representation problem can be generalized to the \emph{partial representation extension} (\textsc{RepExt}) problem in which some arcs of $\M$ are \emph{pre-drawn}, and $\Unit$ must contain these arcs.  \textsc{RepExt} is in turn a special case of the more general \emph{bounded representation} (\BoundRep) problem in which a length $\Len \in \mathbb{Q}$ is given together with lower and upper bounds $\LeftBound(A), \RightBound(A) \in \mathbb{Q}$ for each arc $A$ of $\M$, and the goal is to produce a UIG model $\Unit$ in which all the arcs have length $\Len$ in such a way that $\LeftBound(A) \leq s(A) \leq \RightBound(A)$ for every arc $A$. Here $s(A) \in \mathbb{Q}$ represents the beginning point of $A$.  In this article we consider a further generalization of \BoundRep in which $\Len$, $\LeftBound(A)$, $\RightBound(A)$ are integers, and each beginning point $s(A)$ of $\Unit$ is required to be an integer as well.  We refer to this problem as the \textsc{IntBoundRep}; as far as we know, \textsc{IntBoundRep} has not been considered before.  

\Rep is a classical problem whose research is even older than PIG graphs.  Indeed, \Rep is one of the motivations in the pioneering philosophical work by Goodman~\cite{Goodman1977}, which dates back to the 1940's.  Moreover, Fine and Harrop~\cite{FineHarropJSL1957} developed, in 1957, an effective method to transform a weak mapping of an array (i.e., a PIG model) into a uniform mapping of the same array (i.e., a PIG model of a power of a path); this algorithm is actually the first proof of Robert's PIG=UIG theorem, as far as our knowledge extends.  Linear time algorithms for \Rep are known since more than two decades~\cite{CorneilKimNatarajanOlariuSpragueIPL1995,LinSoulignacSzwarcfiter2009,Mitas1994} and, recently, a logspace implementation has been devised~\cite{KoblerKuhnertVerbitsky2012}.

The research on \textsc{RepExt} and \BoundRep did not begin until recently and, consequently, they are not as studied as \Rep.  We remark that these problems are defined not only for UIG graphs, but for several graph classes with geometric representations.  In the last few years, the partial representation extension and the bounded representation problems were studied for several graph classes~\cite{BalkoKlavikOtachi2013,BlasiusRutter2013,ChaplickDorbecKratochvilMontassierStacho,ChaplickFulekKlavik2013,KlavikKratochvilKrawczykWalczak2012,KlavikKratochvilOtachiRutterSaitohSaumellVyskocil2014,KlavikKratochvilVyskovcil2011}.  Concerning PIG graphs, Balko et al.~\cite{BalkoKlavikOtachi2013} show that the bounded representation problem is solvable in $O(n^2)$ time.  Regarding UIG graphs, Klavík et al.~\cite{KlavikKratochvilOtachiRutterSaitohSaumellVyskocilC2014,KlavikKratochvilOtachiRutterSaitohSaumellVyskocil2014} designed an $O(n^2 + nD)$ time algorithm for \BoundRep, where $D$ is the cost of multiplying large numbers (requiring $r$ bits, where $r$ is the total space consumed by the bounds). As the main open problem, the authors inquire if there exists an algorithm running in less than $O(n^2 + nD)$ time.  In~\cite{KlavikKratochvilOtachiRutterSaitohSaumellVyskocilC2014,KlavikKratochvilOtachiRutterSaitohSaumellVyskocil2014}, a generalization of \BoundRep in which the output UIG model $\Unit$ needs not be equivalent to the input PIG model $\M$ is also considered; what the authors ask is for $G(\Unit)$ to be isomorphic to $G(\M)$.  Whereas \BoundRep is polynomial, this generalization is NP-complete~\cite{KlavikKratochvilOtachiRutterSaitohSaumellVyskocilC2014,KlavikKratochvilOtachiRutterSaitohSaumellVyskocil2014}.   

While introducing their research on \textsc{RepExt}, Klavík et al.\ state that ``specific properties of unit interval representations were \textbf{never} investigated since it is easier to work with combinatorially equivalent proper interval representations''~\cite{KlavikKratochvilOtachiRutterSaitohSaumellVyskocilC2014}.  However, in 1990, Pirlot proved that every PIG graph admits a \emph{minimal} UIG model~\cite{PirlotTaD1990}.  Tough Pirlot's work is not of an algorithmic nature, the main tool he uses is a space efficient representation of PIG models called the \emph{synthetic graph}.  With the aid of an appropriate weighing, this graph reflects the separation constraints that all the equivalent UIG models must satisfy.  As part of his work, Pirlot solves the problem of determining if a PIG graph admits a UIG model in which all the arcs have integer endpoints and a given length $\Len$.  Clearly, this is a specific property of UIG models.  Moreover, Pirlot introduces synthetic graphs to solve the linear program in~\cite[Proposition 5.4]{KlavikKratochvilOtachiRutterSaitohSaumellVyskocil2014} (except for the bound constraints) and, vice versa, the graph used in~\cite[Proposition 5.4]{KlavikKratochvilOtachiRutterSaitohSaumellVyskocil2014} is a synthetic graph (plus two vertices for modeling the bounds).

Similarly as above, in~\cite{GardiDM2007} Gardi claimed that, up to 2007, the algorithm by Corneil et al.~\cite{CorneilKimNatarajanOlariuSpragueIPL1995} was the \textbf{only} one able to solve \Rep in linear time.  Again, by Pirlot's theorem, it makes sense to consider the \emph{minimal UIG representation} (\MinUIG) problem, in which an input PIG model has to be transformed into an equivalent minimal UIG model.  By taking a deeper look to synthetic graphs, Mitas~\cite{Mitas1994} devised a linear time algorithm to solve \MinUIG and, thus, \Rep.  In the present manuscript we show that Mitas' algorithm sometimes fails to find the minimal model.  Yet, her algorithm correctly solves \Rep in linear time. We remark that Mitas' (1994) algorithm is contemporary to the one by Corneil et al.\ (1995).

\MinUIG is implicitly solved in a recent article by Costa et al.~\cite{CostaDantasSankoffXuJBCS2012}, where the authors devise an $O(n^2)$ time and space algorithm to solve the \MinPP problem.  In the \MinPP problem we are given a PIG model $\M$ and the goal is to find a UIG model $\Unit$ representing a power of a path $P_q^k$ in such a way that $\M$ is equivalent to some induced subgraph $\Unit'$ of $\Unit$ and $q,k$ are as minimum as possible.  As proven in~\cite{FineHarropJSL1957,LinRautenbachSoulignacSzwarcfiterDAM2011}, \MinPP is always solvable.  Moreover, $\Unit$ needs not be explicitly constructed, as it is implied by $\Unit'$.  In fact, $\Unit'$ is the solution to \MinUIG, as it follows from~\cite{LinRautenbachSoulignacSzwarcfiterDAM2011} (see also Section~\ref{sec:powers}).  In~\cite[Chapter~9]{Soulignac2010}, Soulignac mentions that Mitas' algorithm can be used to find $\Unit'$ in linear time.  Yet, Costa et al.\ do not mention this fact in~\cite{CostaDantasSankoffXuJBCS2012} although they reference~\cite{Soulignac2010} to explain the strong relation between \Rep and \MinPP.  

In this article we consider the unrestricted, bounded, and minimal representation problems for the broader class of unit circular-arc graphs.  As far as our knowledge extends, only the unrestricted version has been considered, while Lin and Szwarcfiter leave some open problems related to the minimal representation problem~\cite{LinSzwarcfiterSJDM2008}.  

As for PIG graphs, the recognition problem for PCA graphs is solvable in linear time~\cite{KaplanNussbaumDAM2009,SoulignacA2013} or logspace~\cite{KoblerKuhnertVerbitsky2012}.  Again, a PCA model or a forbidden induced subgraph is obtained according to whether the input graph is PCA or not.  We remark, however, that solving the recognition problem for PCA graphs is not enough to solve the recognition problem for UCA graphs, as not every PCA graph is UCA.  In 1974, Tucker showed a characterization by forbidden subgraphs of those PCA graphs that are UCA~\cite{TuckerDM1974}.  His proof yields an effective method to transform a PCA model $\M$ into an equivalent UCA model $\Unit$.  Unfortunately, the extremes of $\Unit$ are not guarantied to be of a polynomial size and, thus, the corresponding representation algorithm cannot be regarded as polynomial.  More than three decades later, in 2006, Durán et al.~\cite{DuranGravanoMcConnellSpinradTuckerJA2006} described how to obtain a forbidden subgraph in $O(n^2)$ time, thus solving the recognition problem.  The representation problem remained unsolved until Lin and Szwarcfiter showed how to transform any PCA model into an equivalent UCA model in linear time~\cite{LinSzwarcfiterSJDM2008}.  Their algorithm, however, does not output a \emph{negative} certificate when the input graph is not UCA.  The problem of finding a forbidden subgraph in linear time was solved by Kaplan and Nussbaum in~\cite{KaplanNussbaumDAM2009}.  Yet, up to this date, there is no \emph{unified} algorithm for solving the transformation problem while providing a negative certificate when the input model has no equivalent UCA models.  In~\cite{KoblerKuhnertVerbitskyC2013}, Köbler et al.\ mention that the representation problem in logspace is still open.

\subsection{Contributions and outline}

Synthetic graphs appeared more than two decades ago, and they are covered in detail in a book by Pirlot and Vincke~\cite[Chapter~4]{PirlotVincke1997}.  Pirlot and Mitas' articles are written in terms of semiorders; their emphasis is on preference modeling and order theory.  This could be, perhaps, the reason why synthetic graphs have gone unnoticed for many researchers in the field of algorithmic graph theory.  In this manuscript we generalize synthetic graphs to PCA models and we apply them to solve \Rep, \BoundRepBoth, \MinUIG (and its generalization \MinUCA), and \MinPP (and its generalization \MinPC) for UCA graphs.  One of our goals is to show that synthetic graphs provide a simpler theoretical ground for understanding PCA models with separation constraints.  For this reason, we re-prove some known theorems or rewrite some known algorithms in terms of synthetic graphs.

The manuscript is organized as follows.  In Section~\ref{sec:preliminaries} we describe the terminology employed.  In Section~\ref{sec:synthetic graph} we introduce synthetic graphs and show how to use them to solve \BoundRep and \textsc{IntBoundRep} in $O(n^2)$ time, improving over the algorithm in~\cite{KlavikKratochvilOtachiRutterSaitohSaumellVyskocilC2014,KlavikKratochvilOtachiRutterSaitohSaumellVyskocil2014} even when restricted to UIG graphs.  In Section~\ref{sec:tucker} we show a new version of Tucker's characterization which implies a linear time representation algorithm with negative certification, thus solving the problem posed in~\cite{KaplanNussbaumDAM2009}.  The implementation of this algorithm appears in Section~\ref{sec:linear algorithm}, while Section~\ref{sec:logspace algorithm} contains a logspace implementation that solves the open problem of~\cite{KoblerKuhnertVerbitskyC2013}.  To apply our algorithm we need to find, as we call it, the ratio of the input model.  This ratio can be computed by invoking the recognition algorithm of~\cite{KaplanNussbaumDAM2009}.  However, we show an alternative implementation in Section~\ref{sec:kaplan and nussbaum}, by taking advantage of synthetic graphs.  The forbidden structure that we employ to characterize UCA graphs is a cycle of the synthetic graph which, a priori, is unrelated to the $(a,b)$-independents and $(x,y)$-circuits employed by Tucker.  In Section~\ref{sec:circuits and independents} we show that, in fact, these structures are strongly related.  In Section~\ref{sec:minimal uca} we extend the concept of minimal models to UCA graphs, and prove that every UCA graph admits a minimal model.  An algorithm to generate such a model in polynomial time is also exhibited.  In Section~\ref{sec:minimal uig}, we consider the \MinUIG problem.  We show that, even though Mitas' algorithm correctly solves \Rep, it sometimes fails to provide a minimal model.  We propose a patch but, unfortunately, the new algorithm runs in $O(n^2)$ time.  In Section~\ref{sec:powers} we show how \MinUIG and \MinUCA can be used so as to solve \MinPP and \MinPC, respectively.  The obtained algorithm for \MinPP runs in $O(n^2)$ time but it consumes only linear space.  Finally, we include some further remarks and open problems in Section~\ref{sec:remarks}.

\subsection{What is linear time for PCA models?}

As discussed in~\cite{Soulignac2010}, every PCA model $\M$ can be encoded with $O(n)$ bits, $n$ being the number of arcs in $\M$.  Thus, in theory, an algorithm on $\M$ is linear when it applies $O(n)$ operations on bits.  However, it is a common practice to assume that $\M$ is implemented with $\Theta(n)$ pointers in such a way that the extremes of an arc can be obtained in $O(1)$ time when the other extreme is given (see~\cite{Soulignac2010}).  Following this tradition, we state that an algorithm is linear when it performs $O(n)$ operations on pointers of size $\Theta(\log n)$.

\section{Preliminaries}
\label{sec:preliminaries}

In this article we consider simple (undirected) graphs and multidigraphs with no loops.  For the sake of simplicity, we refer to the latter as \emph{digraphs} and to its directed edges as \emph{edges}, unless otherwise stated.  For a (di)graph $G$ we write $V(G)$ and $E(G)$ to denote the sets of vertices and bag of edges of $G$, respectively, while we use $n$ and $m$ to denote $|V(G)|$ and $|E(G)|$, respectively.  For any pair $u,v \in V(G)$, we denote the (directed) edge between $u$ and $v$ (\emph{from or starting at} $u$ \emph{to  or ending at} $v$) by $uv$.  This notation is used regardless of whether $uv$ is an edge of $G$ or not.  To avoid confusions, we write $u \to v$ as an equivalent of $uv$ when $G$ is a digraph; the \emph{in-} and \emph{out-degrees} of $v$ are the number of vertices $u$ such that $u \to v$ and $v \to u$, respectively.   

A walk $W$ of a (di)graph $G$ is a sequence of vertices $v_1, \ldots, v_k$ such that $v_iv_{i+1}$ is an edge of $G$, for every $1 \leq i < k$.  Walk $W$ goes \emph{from (or starts at)} $v_1$ \emph{to (or ends at)} $v_k$.  Each walk can be regarded as the bag of edges $\{v_iv_{i+1} \mid 1 \leq i < k\}$.  For the sake of simplicity, we make no distinctions about $W$ begin a sequence of vertices or a bag of edges.  We say that $W$ is a circuit when $v_k = v_1$, that $W$ is a \emph{path} when $v_i \neq v_j$ for every $1 \leq i < j \leq k$, and that $W$ is a \emph{cycle} when it is a circuit and $v_1, \ldots, v_{k-1}$ is a path.  Sometimes we also say that $W$ is a circuit when $v_1 \neq v_k$, to mean that $W, v_1$ is a circuit.  If $G$ contains no cycles, then $G$ is an \emph{acyclic} digraph.

An \emph{edge weighing}, or simply a \emph{weighing}, of a (di)graph $G$ is a function $w\colon E(G) \to \mathbb{R}$.  The value $w(uv)$ is referred to as the \emph{weight} of $uv$ (with respect to $w$).  For any bag of edges, the \emph{weight} of $E$ (with respect to an edge weighing $w$) is $w(E) = \sum_{uv \in E}w(uv)$.  We use two distance measures on a (di)graph $G$ with a weighing $w$.  For $u, v \in V(G)$, we denote by $\Dist{w}(G, u, v)$ the maximum among the weights of the walks from $u$ to $v$, while $\IDist{w}(G, u,v)$ denotes the maximum among the weights of the paths starting at $u$ and ending at $v$.  Note that $\IDist{w}(G,u,v) < \infty$ for every $u,v$, while $\Dist{w}(G, u, v) = \IDist{w}(G, u, v)$ when $G$ contains no cycle of positive weight~\cite{CormenLeisersonRivestStein2009}.  For a weighing $w'$, we write $(\IDist{w} \circ \IDist{w'})(G, u, v) = \max\{w(W) \mid W \text{ is a path from $u$ to $v$ with } w'(W) = \IDist{w'}(G, u, v)\}$.  In other words, $\IDist{w} \circ \IDist{w'}$ measures the $w$-distance from $u$ to $v$ when only those walks that impose the maximum $w'$-distance from $u$ to $v$ are considered.  For the sake of notation, we omit the parameter $G$ when there are no ambiguities.

A \emph{straight plane} (di)graph, or simply a \emph{plane} (di)graph, is a (di)graph whose vertices are coordinates in the plane and whose edges are non-crossing straight lines.  Similarly, a \emph{toroidal} (di)graph is a (di)graph whose vertices and edges can be placed on the surface of a torus in such a way that no pair of edges intersect.

A \emph{proper circular-arc} (PCA) model $\M$ is a pair $(C, \A)$, where $C$ is a circle and $\A$ is a collection of inclusion-free arcs of $C$ such that no pair of arcs in $\A$ cover $C$.  When traversing the circle $C$, we always choose the clockwise direction.  If $s, t$ are points of $C$, we write $(s, t)$ to mean the arc of $C$ defined by traversing the circle from $s$ to $t$; $s$ and $t$ are the \emph{extremes} of $(s, t)$, while $s$ is the \emph{beginning point} and $t$ the \emph{ending point}.  For $A \in \A$, we write $A = (s(A), t(A))$. The \emph{extremes} of $\A$ are those of all arcs in $\A$, and two extremes $s_1s_2$ of $\M$ are \emph{consecutive} when there is no extreme $s \in (s_1, s_2)$ (note that $s_2s_1$ is not consecutive in this case).  We assume $C$ has a special point $0$ that is used for describing the bounds on the extremes (cf.\ below).  This point is only denotational for the unbounded case.  For every pair of points $p_1, p_2$, we write $p_1 < p_2$ to indicate that $p_1$ appears before $p_2$ in a traversal of $C$ from $0$.  Similarly, we write $A_1 < A_2$ to mean that $s(A_1) < s(A_2)$ for any pair of arcs $A_1, A_2$ on $C$.  

A \emph{unit circular-arc} (UCA) model is a circular-arc model $\M$ in which all the arcs have the same length.  Let $A_1 < \ldots < A_n$ be the arcs of $\M = (C, \A)$, $\Circ, \Len \in \mathbb{Q}_{> 0}$, $\Thres, \BegDist \in \mathbb{Q}_{\geq 0}$, and $\LeftBound, \RightBound \colon \A \to \mathbb{Q}_{\geq0}$.  We say that $\M$ is a \emph{$(\Circ, \Len, \Thres, \BegDist, \LeftBound, \RightBound)$-CA model} when:
\begin{enumerate}[({unit}$_1$)]
  \item $C$ has circumference $\Circ$, \label{def:unit-1}
  \item all the arcs of $\A$ have length $\Len$, \label{def:unit-2}
  \item $(p_1, p_2)$ has length at least $\Thres$ for every pair of consecutive extremes $p_1p_2$, \label{def:unit-3}
  \item $(s_1, s_2)$ has length at least $\Thres+\BegDist$ for any pair of beginning points $s_1, s_2$, and \label{def:unit-3s}
  \item $\LeftBound(A_i) \leq s(A_i) \leq \Circ - \RightBound(A_i)$ for every $1 \leq i \leq n$. \label{def:unit-4}
\end{enumerate}
\newcommand{\unitref}[1]{\itemref{unit}{unit}{#1}}
Intuitively, $\M$ is a UCA model in which the extremes are separated by at least $\Thres$ space, the beginning points are separated by $\Thres + \BegDist$ space, and $\LeftBound(A_i)$ and $\RightBound(A_i)$ are lower bounds of the separation from $0$ to $s(A_i)$ and from $s(A_i)$ to $0$, respectively.  We simply write that $\M$ is a \emph{$(\Circ, \Len, \Thres, \BegDist)$-CA model} to indicate that $\LeftBound = \RightBound = 0$, and that $\M$ is a $(\Circ, \Len)$-CA model to mean that $\M$ is a $(\Circ, \Len, 1, 0)$-CA model.  To further simplify the notation, we refer to the tuple $\Descriptor = (\Circ, \Len, \Thres, \BegDist, \LeftBound, \RightBound)$ as a \emph{UCA descriptor}, and we say that $\Descriptor$ is \emph{integer} when $\Circ$, $\Len$, $\Thres$, $\BegDist$, $\LeftBound$, and $\RightBound$ are integers.  Similarly, a $\Descriptor$-CA model $\M$ is \emph{integer} when $\Circ$, $\Len$ and all the extremes of $\M$ are integers.  

A \emph{proper interval} (PIG) model is a PCA model $\M$ in which no arc crosses $0$; if $\M$ is also UCA, then $\M$ is a \emph{unit} interval (UIG) model.  Any UIG model $\M$ is a $\Descriptor$-CA model for some large enough $\Circ$; for simplicity, we just write $\Circ = \infty$ in this case.  For this reason, we say that $\M$ is an $(\Len, \Thres, \BegDist, \LeftBound, \RightBound)$-IG (resp.\ $(\Len, \Thres, \BegDist)$, $\Len$-IG) model when $\M$ is a $(\infty, \Len, \Thres, \BegDist, \LeftBound, \RightBound)$-CA (resp.\ $(\infty, \Len, \Thres, \BegDist)$, $(\infty, \Len)$-CA) model.  That is, $\M$ is an $(\Len, \Thres, \BegDist, \LeftBound, \RightBound)$-IG model when all the arcs have length $\ell$, every pair of consecutive extremes is separated by $\Thres$ space, every pair of beginning points is separated by $\Thres+\BegDist$ space, and $\LeftBound$ and $\RightBound$ impose lower and upper bounds on the beginning points of $\M$.

Each PCA model $\M$ \emph{represents} a \emph{proper circular-arc} graph $G(\M)$ that contains a vertex for each arc of $\M$ where two vertices are adjacent if and only if their corresponding arcs have nonempty intersection. Conversely, we say that a graph $G$ \emph{admits} a PCA model $\M$ to mean that $G$ is isomorphic to $G(\M)$.  If $\M$ is UCA, then $G(\M)$ is a \emph{unit circular-arc (UCA)} graph, while if $\M$ is PIG (resp.\ UIG), then $G(\M)$ is a \emph{proper interval (PIG)} (resp.\ \emph{unit interval}; \emph{UIG}) graph.  

Clearly, two PCA models $\M_1 = (C_1, \A_1)$ and $\M_2 = (C_2, \A_2)$ are equal when $C_1 = C_2$ and $\M_1 = \M_2$.  We say that $\M_1$ is \emph{equivalent} to $\M_2$ when the extremes of $\M_1$ appear in the same order as in $\M_2$.  Formally, $\M_1$ and $\M_2$ are equivalent if there exists $f\colon\A_1 \to \A_2$ such that $e(f(A))e'(f(B))$ are consecutive if and only if $e(A)e'(B)$ are consecutive, for $e, e' \in \{s,t\}$. By definition, $\M_1$ and $\M_2$ are equivalent whenever they are equal.

In this manuscript we consider several related recognition problems.  In the \emph{representation} (\Rep) problem a UCA model equivalent to an input PCA model $\M$ must be generated.  Of course, \Rep is unsolvable when $\M$ is equivalent to no UCA model, a \emph{negative certificate} is desired in such a case.  In the \uRep problem, a (an integer) UCA descriptor $\Descriptor$ is given together with $\M$, and the goal is to build a (an integer) $\Descriptor$-CA model $\Unit$.  We remark that an integer $\Unit$ equivalent to $\M$ exists whenever $\Descriptor$ is integer and \uRep is solvable.  The \emph{bounded representation} (\BoundRep) is a slight variation of \uRep in which a feasible $\Thres > 0$ must be found by the algorithm, as it is not given as input.  That is, we are given a PCA model $\M = (C, \A)$ together with $\Circ, \Len \in \mathbb{Q}_{>0}$, $\BegDist \in \mathbb{Q}_{\geq 0}$ and $\LeftBound, \RightBound \colon \A \to \mathbb{Q}_{\geq 0}$, and we ought to find a $\Descriptor$-CA model equivalent to $\M$ for some UCA descriptor $\Descriptor = (\Circ, \Len, \Thres, \BegDist, \LeftBound, \RightBound)$ with $\Thres \in \mathbb{Q}_{> 0}$.  The \emph{integer bounded representation} (\IntBoundRep) problem is a generalization of \BoundRep in which all the input values are integers and the output model must be integer as well.  We also study the \MinUIG, \MinUCA, \MinPP, and \MinPC problems that are related to \emph{minimal} models.  We postpone their definitions to Sections \ref{sec:minimal uca}~and~\ref{sec:powers}.

\subsection{Restrictions on the input models}

As it is customary in the literature, in this article we assume that all the arcs of a PCA model $\M$ are open and no two extremes of $\M$ coincide.  The reason behind these assumptions is that $\M$ can always be transformed into an equivalent model $\M'$ that satisfies these properties.  A word of caution is required, though, as in this article we deal with the lengths of the arcs.  If we allow coincidences in the extremes of $\M$, for instance, it is possible to shrink the length of the arcs or the circle of some UCA models.  We emphasize, nevertheless, that all the arguments in this article, with the obvious adjustments, work equally well without these assumptions.  In particular, note that the articles by Klavík et al., Mitas, and Pirlot allow coincident extremes~\cite{KlavikKratochvilOtachiRutterSaitohSaumellVyskocilC2014,Mitas1994,PirlotTaD1990,PirlotDAM1991}.  

By definition, for us PCA models cannot have two arcs covering the circle.  This is a somehow artificial restriction that we impose for the sake of simplicity.  In general, this class of models is said to be \emph{normal}.  However, it is well known that every non-normal PCA model can be transformed into a normal PCA model in linear time or logspace (see e.g.~\cite{KaplanNussbaumDAM2009}). Moreover, note that if two arcs in a UIG model cover the circle, then such a model represents a complete graph.  The complete graph on $n$ vertices admits the \emph{minimal} UIG model $\{(i, i+n+1) \mid 1 \leq i \leq n\}$, thus we do not lose much by excluding these non-normal models when dealing with \Rep, \MinUCA, and \MinUIG.  In turn, the fact that $\M$ is normal is not used in Theorem~\ref{thm:separation constraints}, thus \BoundRepBoth is also solvable for non-normal models.

Finally, we require two additional restrictions on the input PCA models for technical reasons.  We say that a PCA model $\M$ with arcs $A_1 < \ldots < A_n$ is \emph{trivial} when either 
\begin{enumerate}
  \item $s(A_n) < t(A_1)$, or 
  \item $s(A_i)t(A_i)$ are consecutive for some $1 \leq i \leq n$.
\end{enumerate}
If 1.\ holds, then we cannot claim that $h(\M) \geq 1$ in Section~\ref{sec:synthetic graph:separation}.  However, in this case $\M$ represents a complete graph and $\{(i, i+n+1) \mid 1 \leq i \leq n\}$ is the unique minimal and integer UCA model equivalent to $\M$.  Thus all the considered problems are trivial in this case.  If 2.\ is true, then $A_i \to A_i$ is a loop of the digraph $\BoundedSyn(\M)$ defined in Section~\ref{sec:synthetic graph}.  We can certainly allow the existence of such a loop in $\BoundedSyn$.  However, this edge plays no role in the considered problems as $\SEP(A_i \to A_i) < 0$ by \sepref{3}.

\section{The synthetic graph of a PCA model}
\label{sec:synthetic graph}

Pirlot introduced the synthetic graph of a PIG model~\cite{PirlotTaD1990,PirlotDAM1991} to represent the separation constraints of its extremes in any equivalent UIG representation.  In this section we extend them to PCA models and we show that they correctly reflect the separation constraints in any equivalent UCA model.

Let $\M = (C, \A)$ be a PCA model with arcs $A_1 < \ldots < A_n$.  The \emph{bounded synthetic graph} of $\M$ is the digraph $\BoundedSyn(\M)$ (see Figure~\ref{fig:synthetic graph}) that has a vertex $v(A_i)$ for each $A_i \in \A$ and a vertex $A_0$, and whose edge set is $E_\Steps \cup E_\Noses \cup E_\Hollows \cup E_\Bounds$, where:
\begin{itemize}
 \item $E_\Steps = \{v(A_i) \to v(A_{i+1}) \mid 1 \leq i \leq n,\ A_{n+1} = A_1\}$,
 \item $E_\Noses = \{v(A_i) \to v(A_j) \mid t(A_i)s(A_j) \text{ are consecutive in } \M\}$,
 \item $E_\Hollows = \{v(A_i) \to v(A_j) \mid s(A_i)t(A_j) \text{ are consecutive in } \M\}$, and
 \item $E_\Bounds = \{A_0 \to v(A_i), v(A_i) \to A_0 \mid 1\leq i \leq n\}$.
\end{itemize}
The edges in $E_\Steps$, $E_\Noses$, $E_\Hollows$, and $E_\Bounds$ are said to be the \emph{steps}, \emph{noses}, \emph{hollows}, and \emph{bounds} of $\BoundedSyn(\M)$, respectively.  (Note that $E_\Steps$, $E_\Noses$ and $E_\Hollows$ could have a nonempty intersection, even if this is not the common case.  However, $\BoundedSyn(\M)$ has no loops as $\M$ is not trivial.)  For the sake of simplicity, we usually drop the parameter $\M$ from $\BoundedSyn(\M)$ when no ambiguities are possible.  Moreover, we regard the arcs of $\M$ as being the vertices of $\BoundedSyn$, thus we may say that $A_i \to A_j$ is a nose instead of writing that $v(A_i) \to v(A_j)$ is a nose.  

\newlength{\VerticalCentering}
\newlength{\WidthA}
\newlength{\WidthB}
\settoheight{\VerticalCentering}{\includegraphics{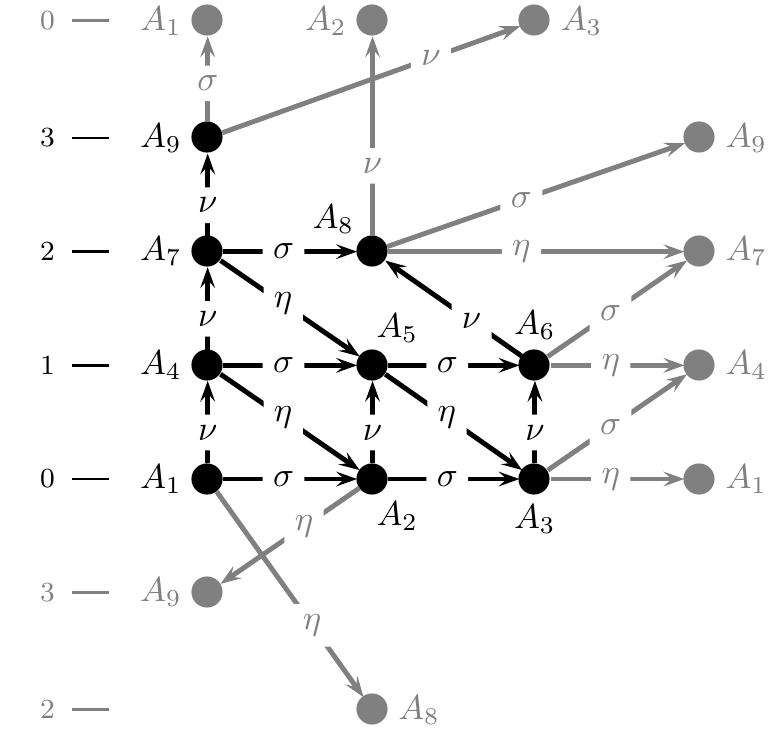}}
\settowidth{\WidthA}{\includegraphics{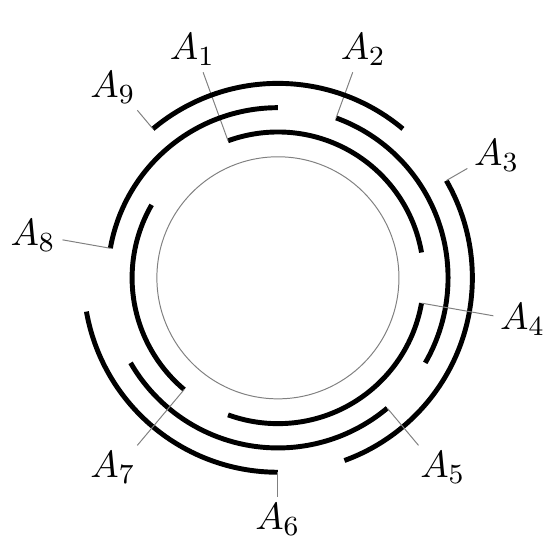}}
\settowidth{\WidthB}{\includegraphics{\figurepath fig-synthetic-graph}}

\begin{figure}[th]
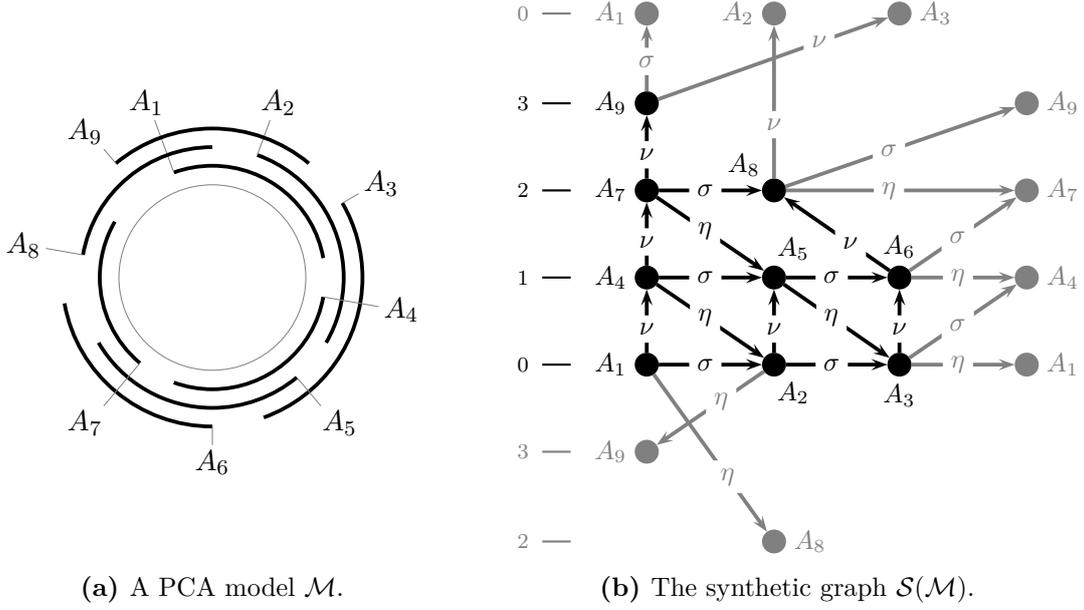
%
  \hfil
  \begin{minipage}[c]{\WidthA}
    \begin{minipage}[c]{0pt}
      \rule{0pt}{\VerticalCentering}
    \end{minipage}%
    \begin{minipage}[c]{\linewidth}
      \includegraphics{\figurepath fig-pca-model}
    \end{minipage}%
    \subcaption{A PCA model $\M$.}%
  \end{minipage}%
  \hfil
  \begin{minipage}[c]{\WidthB}
    \includegraphics{\figurepath fig-synthetic-graph}
    \subcaption{The synthetic graph $\Syn(\M)$.}
  \end{minipage}%
  \hfil
  \caption{\small Synthetic graph $\BoundedSyn \setminus A_0$ of a PCA model. Each gray vertex corresponds to a black vertex (we separate them for the sake of exposition) and each edge is drawn only once.  The edges are labeled with $\Noses$, $\Hollows$, and $\Steps$ according to whether they are noses, hollows, and steps, respectively.  The height of $\M$ is $\Height = 3$ and each vertex is drawn in a row that corresponds to its height; the height is indicated to the left.  Note that there are $1$-, $(-\Height)$-, and $(1-\Height)$-noses, $0$-, $(-1)$-, $\Height$-, and $(\Height-1)$-hollows, and $0$-, $1$-, and $(-\Height)$-steps.}\label{fig:synthetic graph}
\end{figure}

We classify the edges of $\BoundedSyn$ in two classes according to the positions of their arcs.  We say a step (resp.\ nose) $A_i \to A_j$ is \emph{internal} when $i < j$, while a hollow is \emph{internal} when $i > j$.  Non-internal edges are referred to as \emph{external}; in particular, all the bounds are external.  Observe that every step is internal except $A_n \to A_1$.  Similarly, a nose $A_i \to A_j$ is internal if and only if the arc $(s(A_i), s(A_j))$ does not cross $0$, while a hollow $A_i \to A_j$ is internal if and only if $(s(A_j), s(A_i))$ does not cross $0$.  Since the purpose of $A_0$ is to represent the point $0$ of $\M$, we can say, in short, that $A_i \to A_j$ is internal when $0$ is not crossed in the traversal of the extremes involved in the definition of $A_i \to A_j$.

We define a special edge weighing $\SEP_\Descriptor$ of $\BoundedSyn$ whose purpose is to indicate how far or close must $s(A_i)$ and $s(A_j)$ be in any $\Descriptor$-CA model equivalent to $\M$, for every edge $A_i \to A_j$ of $\BoundedSyn$.  For a UCA descriptor $\Descriptor$, the edge weighing \emph{$\SEP_\Descriptor$} is such that:
\begin{enumerate}[($\SEP_1$)]
 \item $\SEP_\Descriptor(A_i \to A_j) = \Thres+\BegDist-\Circ q$ if $A_i \to A_j$ is a step, \label{def:sep-1}
 \item $\SEP_\Descriptor(A_i \to A_j) = \Thres+\Len-\Circ q$ if $A_i \to A_j$ is a nose,\label{def:sep-2}
 \item $\SEP_\Descriptor(A_i \to A_j) = \Thres + \Circ q-\Len$ if $A_i \to A_j$ is a hollow, and\label{def:sep-3}
 \item $\SEP_\Descriptor(A_0 \to A_i) = \LeftBound(A_i)$ and $\SEP(A_i \to A_0) = \RightBound(A_i) - \Circ$ for every $1 \leq i \leq n$,\label{def:sep-4}
\end{enumerate}
where $q \in \{0,1\}$ equals $0$ if and only if $A_i \to A_j$ is internal.  For the sake of notation, we omit the subscript $\Descriptor$ from $\SEP$ when no ambiguities are possible.  Suppose for a moment that $\M$ is a $\Descriptor$-CA model.  By definition, $s(A_{j}) \geq s(A_{i})+\Len+\Thres$ when $A_i \to A_j$ is a internal nose of $\BoundedSyn$, while $s(A_{j}) \geq s(A_{i}) + \Len + \Thres - \Circ$ when $A_i \to A_j$ is an external nose of $\BoundedSyn$.  Thus, equation \sepref{2} models the non-intersection constraints imposed by the noses of $\BoundedSyn$.  A similar analysis shows that \sepref{1} indicates that all the beginning points must be at distance at least $\Thres+\BegDist$, \sepref{3} models the intersection constraints imposed by the hollows of $\BoundedSyn$, and \sepref{4} models the bound constraints, assuming that $A_0$ represents $0$ in $\M$.  

As we shall see in Theorem~\ref{thm:separation constraints}, a $\Descriptor$-CA model equivalent to $\M$ exists when the longest path problem with weight $\SEP$ has a feasible solution on $\BoundedSyn$.  In such case, a $\Descriptor$-CA model can be generated by observing the distances from $A_0$.  With this in mind, we define $\Unit(\M, \Descriptor)$ to be the $\Descriptor$-CA model with arcs $U_1, \ldots, U_n$ such that $s(U_i) = \IDist{\SEP}(A_0, A_i)$, for every $1 \leq i \leq n$ (we assume arithmetic modulo $\Circ$).  For simplicity, we omit $\M$ and $\Descriptor$ from $\Unit$ as usual. 

\begin{theorem}\label{thm:separation constraints}
 The following statements are equivalent for a PCA model $\M$ with arcs $A_1 < \ldots < A_n$ and a (an integer) UCA descriptor $\Descriptor$:
 \begin{enumerate}[(i)]
  \item $\M$ is equivalent to a $\Descriptor$-CA model.\label{thm:separation constraints:equivalence}
  \item $\SEP(\W) \leq 0$ for every cycle $\W$ of $\BoundedSyn$. \label{thm:separation constraints:cycles}
  \item $\Unit$ is a (an integer) $\Descriptor$-CA model equivalent to $\M$.\label{thm:separation constraints:model}
 \end{enumerate}
\end{theorem}

\begin{proof}
 (\ref{thm:separation constraints:equivalence}) $\Rightarrow$ (\ref{thm:separation constraints:cycles}).  Suppose $\M$ is equivalent to a $\Descriptor$-CA model $\M'$ with arcs $A_1', \ldots, A_n'$ such that $A_i'$ corresponds to $A_i$ for $1 \leq i \leq n$.  Write $s(A_0')$ to mean the point $0$ of $\M'$.  Then, it is not hard to see (cf.\ above) that $s(A'_{j}) \geq s(A_i') + \SEP(A_i \to A_{j})$ for every edge $A_i \to A_j$ of $\BoundedSyn$.  Hence, by induction, $s(A_i') \geq s(A_i') + \SEP(\W)$ for every cycle $\W$ of $\BoundedSyn$ that contains $A_i$.
 
 (\ref{thm:separation constraints:cycles}) $\Rightarrow$ (\ref{thm:separation constraints:model}).  Let $U_1 < \ldots < U_n$ be the arcs of $\Unit$, $U_{n+1} = U_1$, $A_{n+1} = A_1$, and note that $\Dist{\SEP}(A_i, A_j) = \IDist{\SEP}(A_i, A_j)$ for every $0 \leq i, j \leq n$ as $\BoundedSyn$ has no cycles of positive length.  Thus, by \sepref{4}, $\Unit$ satisfies \unitref{4} as $s(U_i) = \IDist{\SEP}(A_0, A_i) \geq \LeftBound(A_i)$ and $s(U_i) + \RightBound(A_i) - \Circ = \IDist{\SEP}(A_0, A_i) + \RightBound(A_i) - \Circ \leq \IDist{\SEP}(A_0, A_0) = 0$.  Since $\Unit$ satisfies \unitref{1}--\unitref{2} by definition, it follows that $\Unit$ is a $(\Circ, \Len, \Thres', \BegDist', \LeftBound, \RightBound)$-CA model for some $\Thres', \BegDist'$.  To prove that $\Unit$ is a $\Descriptor$-CA model equivalent to $\M$, it suffices to see that (a) $s(U_{i}) + \Thres + \BegDist \leq s(U_{i+1})$ for every $1 \leq i \leq n$, (b)  $s(U_i) + \Thres \leq t(U_j)$ when $s(A_i)t(A_j)$ are consecutive in $\M$, and (c) $t(U_i) + \Thres \leq s(U_j)$ when $t(A_i)s(A_j)$ are consecutive in $\M$.
 \begin{description}
  \item[(a)]  $A_{i} \to A_{i+1}$ is a step, thus $s(U_{i}) + \Thres + \BegDist= \IDist{\SEP}(A_0, A_{i}) + \Thres + \BegDist \leq \IDist{\SEP}(A_0, A_{i+1}) = s(U_{i+1})$.
  \item[(b)]  $A_i \to A_j$ is a hollow of $\BoundedSyn$; let $q \in \{0,1\}$ be $1$ if and only if $(s(A_j), s(A_i))$ crosses $0$.  Note that, equivalently, $q = 1$ if and only if $A_i \to A_j$ is external.  Thus, $t(U_j) = s(U_j) + \Len - \Circ q = \Dist{\SEP}(A_0, A_j) + \Len - \Circ q \geq \Dist{\SEP}(A_0, A_i) + \SEP(A_i \to A_j) + \Len - \Circ q = s(U_i) + \Thres$.  
  \item[(c)]  $A_i \to A_j$ is a nose of $\BoundedSyn$; if $q \in \{0,1\}$ equals $1$ when $(s(A_i), s(A_j))$ crosses $0$, then $s(U_j) \geq \Dist{\SEP}(A_0, A_i) + 1 + \Len- \Circ q \geq t(U_i) + \Thres$.
 \end{description}
 
 (\ref{thm:separation constraints:model})  $\Rightarrow$ (\ref{thm:separation constraints:equivalence}). Trivial.
\end{proof}

When restricted to PIG models, Theorem~\ref{thm:separation constraints} is a somehow alternative formulation of Proposition 2.5 in~\cite{PirlotTaD1990}; see also Proposition~5.4 in~\cite{KlavikKratochvilOtachiRutterSaitohSaumellVyskocilC2014}.

\subsection{The bounded representation problem}
\label{sec:synthetic graph:bounded algorithm}

Though simple enough, Theorem~\ref{thm:separation constraints} allows us to solve \uRep as follows.  First, we build the digraph $\BoundedSyn$ in which every edge $A_i \to A_j$ is weighed with $s_{ij} = \SEP_\Descriptor(A_i, A_j)$.  Then, we invoke the Bellman-Ford shortest path algorithm~\cite{CormenLeisersonRivestStein2009} on $\BoundedSyn$ to obtain $s_{i} = \Dist{\SEP}(A_0, A_i)$ for every $0 \leq i \leq n$.  If Bellman-Ford ends in success, then we output $\Unit(\M, \Descriptor)$; otherwise, we output the cycle of positive weight found as the negative certificate.  

Bellman-Ford computes each value $s_i = s_i^n$ in an iterative manner.  At iteration $k$, the value of $s_i$ is updated to $s_i^k = \max\{s_j^{k-1} + s_{ji} \mid A_j \to A_i \in E(\BoundedSyn)\}$ for every $0 \leq i \leq n$.  As $\BoundedSyn$ has $O(n)$ edges, a total of $O(n^2)$ arithmetic operations are performed.  By \sepref{1}--\sepref{4}, $s_i^k \in \mathbb{N}$ when $\Descriptor$ is integer, thus $O(1)$ time is required by each operation.  However, as it was noted by Klavík et al.~\cite{KlavikKratochvilOtachiRutterSaitohSaumellVyskocilC2014}, we cannot assume $O(1)$ time per operation when $\Descriptor$ is non-integer.  The inconvenient is that to compare two fractional values $a_1/b_1$ and $a_2/b_2$ we have to multiply them with a common multiple of $b_1$ and $b_2$.  Thus, a priori, the number of bits used to represent $s_j^{k-1}$ could be large, and the operations required to compute $s_i^k$ could take more than constant time.  

It turns out that we can represent $s_i$ with $O(1)$ bits in a simple manner.  The idea is to use a \emph{distance tuple} $t = \langle b, [\Circ], [\Len], [\Thres], [\BegDist] \rangle$ with $b \in \{\LeftBound(A), \RightBound(A) \mid A \in \A\} \cup \{-\infty\}$, $[\Circ], [\Len] \in \mathbb{Z}$, and $[\Thres], [\BegDist] \in \mathbb{N}$, in order to represent the rational
\[[ t ] = b + [\Circ]\Circ + [\Len]\Len + [\Thres]\Thres + [\BegDist]\BegDist.\]
So, for instance, we can represent $s_{ij} = \SEP(A_i \to A_j)$ as in the following table, where $q \in \{0,1\}$ equals $0$ if and only if $A_i \to A_j$ is internal.
\begin{center}
\begin{tabular}{l|c|c|c|c|c}
  Type          & $b$                  & $[\Circ]$    & $[\Len]$    & $[\Thres]$    & $[\BegDist]$    \\\hline
  Step          & 0                    & $-q$         & 0           & $1$           & $1$             \\
  Nose          & $0$                  & $-q$         & $1$         & $1$           & $0$             \\
  Hollow        & $0$                  & $q$          & $-1$        & $1$           & $0$             \\
  Bound ($i=0$) & $\LeftBound(A_i)$    & $0$          & $0$         & $0$           & $0$             \\
  Bound ($j=0$) & $\RightBound(A_i)$   & $-1$         & $0$         & $0$           & $0$             \\
\end{tabular}
\end{center}
Analogously, we implement $s_i = \Dist{\SEP}(A_0, A_i)$ using a distance tuple for every $1 \leq i \leq n$.  Just note that if $s_0$ is ever updated, then $\BoundedSyn$ has a cycle of positive weight, thus we can immediately halt Bellman-Ford in failure.  By doing so we observe, by invariant, that $s_i^k = [t]$ for every iteration $k$ and some distance tuple $t = \langle [\Circ], [\Len], [\Thres], [\BegDist], b \rangle$  in which $-n \leq [\Circ], [\Len] \leq n$, $0 \leq [\Thres], [\BegDist] \leq n$.

With the above implementation, each arithmetic operation performed by Bellman-Ford costs $O(1)$ time, as it involves only $O(1)$ of the input values.  We conclude, therefore, that \uRep can be solved in $O(n^2)$ time, even when $\Descriptor$ is non-integer.  As far as our knowledge extends, this is the first polynomial algorithm to solve this problem.

By definition, \BoundRepBoth is solvable if and only if \uRep is solvable for some (integer) UCA descriptor $\Descriptor$.  The main difference between both problems is that $\Thres$ is an input of \uRep whereas a feasible $\Thres$ must be found by \BoundRep.  A simple solution for \BoundRepBoth is to invoke the above algorithm with a small enough value of $\Thres$.  For instance, if $a_1/b_1, \ldots, a_k/b_k$ are the bounds of $\LeftBound$ and $\RightBound$, then we can take $\Thres = \frac{1}{B}$ where $B = \prod_{i=1}^k b_i$.  In other words, we transform every weight of $\SEP$ into an integer before invoking Bellman-Ford in the algorithm above.  This algorithm is efficient when $\Thres$ consumes $O(1)$ bits, e.g., when $\Descriptor$ is integer.  But, it is not efficient in the general case as to compare two distance tuples we need to operate with $\Thres$.

An alternative solution for \BoundRep is to find any $\Thres < \Thres_*$, where $\Thres_*$ is the maximum such that $\M$ is equivalent to a $(\Circ,\Len,\Thres_*,\BegDist,\LeftBound,\RightBound)$-CA model. To find $\Thres$ we invoke the algorithm for \uRep, but instead giving an input number for $\Thres$, we just think of $\Thres$ as a placeholder for a value lower than or equal to $\Thres_*$.  As before, $s_{ij} = \SEP(A_i \to A_j)$ and $s_i = \Dist{\SEP}(A_0, A_i)$ are encoded with distance tuples $[t_{ij}]$ and $[t_i]$, respectively.  However, we re-implement the comparison operator to cope with the fact that $\Thres$ is an indeterminate value.  

Let $[\Thres_i]$ and $[\Thres_{ij}]$ be the coefficients of $t_{ij}$ and $t_i$ that multiply $\Thres$, respectively.  For a distance tuple $t =  \langle b, [\Circ], [\Len], [\Thres], [\BegDist] \rangle$, let
\[\lVert t \rVert = [t] - [\Thres]\Thres = b + [\Circ]\Circ + [\Len]\Len + [\BegDist]\BegDist.\]
The main observation is that $\Dist{\SEP}(A_0, A_i) < \Dist{\SEP}(A_0, A_j) + \SEP(A_j \to A_i)$ if and only if 
\begin{itemize}
  \item $\lVert s_i \rVert < \lVert s_{ji} \rVert + \lVert s_j \rVert$, or
  \item $\lVert s_i \rVert = \lVert s_{ji} \rVert + \lVert s_j \rVert$ and $[\Thres_i] < [\Thres_{j}] + [\Thres_{ji}]$.
\end{itemize}
Then, every arithmetic operation costs $O(1)$ time as it involves only $O(1)$ input values.

If Bellman-Ford ends in success, then $\Thres = \min\{\Dist{\SEP}(A_0, A_j) - \Dist{\SEP}(A_0, A_i) \mid A_i \to A_j \in E(\BoundedSyn)\} = \min\{[t_i] - [t_j] \mid A_i \to A_j \in E(\BoundedSyn)\} > 0$ can be obtained in $O(n)$ time.  By Theorem~\ref{thm:separation constraints}, the algorithm is correct as $\SEP(\W) \leq 0$ for every cycle $\W$ of $\BoundedSyn$.  As for the certification problem, note that $\Thres$ consumes $O(1)$ bits, thus we can output $\Unit(\M)$ in $O(n)$ time.  Moreover, any cycle of positive weight found by Bellman-Ford can be used for negative certification. 

\begin{theorem}
  \BoundRep, \textsc{IntBoundRep}, and \uRep can be solved in $O(n^2)$ time and $O(n)$ space.
\end{theorem}

\subsection{The separation of a boundless walk}
\label{sec:synthetic graph:separation}

The cycles of $\BoundedSyn$ with maximum $\SEP$-values play a fundamental role when deciding if $\M$ admits an equivalent $\Descriptor$-CA model, as shown in Theorem~\ref{thm:separation constraints}.  The purpose of this section is to analyze how do these separations look like in the boundless synthetic graph.  The \emph{(boundless) synthetic graph} of $\M$ is just $\Syn(\M) = \BoundedSyn(\M) \setminus A_0$; for the sake of simplicity, we drop the parameter $\M$ as usual.  The main tool that we apply is a pictorial description of $\Syn$, that generalizes the work of Mitas~\cite{Mitas1994} on PIG models (see Section~\ref{sec:minimal uig}).  Roughly speaking, Mitas arranges the vertices of $\Syn$ into a matrix, where the row and column of $A_i$ correspond to its height (cf.\ below) and the number of internal hollows of some paths from $A_1$, respectively.  

Let $\M = (C, \A)$ be a PCA model with arcs $A_1 < \ldots < A_n$.  The \emph{height} $\Height(A_i)$ of $A_i$ ($1 \leq i \leq n$) is recursively defined as follows:
\[
 \Height(A_i) =
  \begin{cases}
   0 & \text{if } s(A_i) < t(A_1) \\
   1 + \Height(A_j)       & \text{otherwise, where } A_j \text{ is the maximum such that } t(A_j) < s(A_i).
  \end{cases}
\]
The \emph{height} of $\M$ is defined as $\Height(\M) = \Height(A_n)$; note that $h(\M) \geq 1$ (because $\M$ is not trivial).  For the sake of notation, we drop the parameter $\M$ as usual.  In Figure~\ref{fig:synthetic graph}, the vertices are drawn in levels according to their height.

It is important to note that internal noses and steps ``jump'' to higher or equal levels, while external noses and steps ``jump'' to lower levels.  Similarly, internal hollows ``jump'' to equal or lower levels while external hollows ``jump'' to higher levels.  We need a more explicit description of how does the height change when an edge is traversed.  In general, we say that $A_i \to A_j$ \emph{has \Jump\ jump} for $\Jump = \Height(A_j) - \Height(A_i)$.  For the sake of notation, we refer to noses (resp.\ steps, hollows) with \Jump\ jump simply as \Jump-noses (resp.\ \Jump-steps, \Jump-hollows).

It is not hard to see (check Figure~\ref{fig:synthetic graph}) that $\Syn$ has three kinds of noses, namely $1$-, $(-\Height)$-, and $(1-\Height)$-noses.  Moreover, if $A_i \to A_j$ is either a $(-\Height)$- or $(1-\Height)$-nose, then $\Height(A_j) = 0$.  Similarly, there are three kinds of steps, namely $0$-, $1$-, and $(-\Height)$-steps, and four kinds of hollows, namely $0$-, $(-1)$-, $\Height$-, and external $(\Height-1)$-hollows.  Note that we need to differentiate between internal $0$-hollows and external $(\Height-1)$-hollows when $\Height = 1$.  For the sake of simplicity, we will refer to $A_i \to A_j$ as an \emph{$(\Height-1)$-hollow} to mean that $A_i \to A_j$ is an \textbf{external} $(\Height-1)$-hollow.  We emphasize that no confusions are possible because $\M$ has no external $0$-hollows; otherwise $A_i$ and $A_j$ would cover the circle of $\M$.  Observe that, as it happens with noses, $\Height(A_i) = 0$ for every $\Height$- or $(\Height-1)$-hollow $A_i \to A_j$, while $\Height(A_1) = 0$ for the unique $(-\Height)$-step $A_n \to A_1$.  Obviously, the jump of a walk $\W$ depends exclusively on the number of different kinds of noses, hollows and steps that it contains.  We write $\Noses_\Jump(\W)$, $\Hollows_\Jump(\W)$, and $\Steps_\Jump(\W)$ to indicate the number of $\delta$-noses, $\delta$-hollows, and $\delta$-steps of $\W$, respectively.  As usual, we do not write the parameter $\W$ when it is clear from context.  The following observation describes the jump of $\W$.

\begin{observation}
 If $\W$ is a walk from $A_i$ to $A_j$ in $\Syn$, then 
 \begin{equation}
  \Height(A_j) - \Height(A_i) = \Noses_1 + \Steps_1 - \Hollows_{-1} + \Height(\Hollows_\Height - \Noses_{-\Height} - \Steps_{-\Height}) + (\Height-1)(\Hollows_{\Height-1} - \Noses_{1-\Height}) \label{eq:walk jump}
 \end{equation}
\end{observation}

We now define two kinds of walks that are of particular interest for us.  These walks correspond to what Tucker calls by the names of \emph{$(a,b)$-independent} and \emph{$(x,y)$-circuits} of a PCA model (see~\cite{TuckerDM1974} and Section~\ref{sec:circuits and independents}).  We say that a walk of $\Syn$ is a \emph{nose walk} when it contains no hollows, while it is a \emph{hollow walk} when it contains no noses and $\Hollows_h + \Hollows_{h-1} \geq \Steps_{-h}$.  Note that a walk is both a nose and a hollow walk only if all its edges are steps; in general, walks that contain only steps are referred to as \emph{step walks}.  Nose and hollow walks are important because they impose lower and upper bounds for the circumference of the circle in a UCA model.

By Theorem~\ref{thm:separation constraints}, if $\M$ is equivalent to a $\Descriptor$-CA model, then $\SEP(\W_N) \leq 0$ for every nose cycle $\W_N$.  By definition,
\begin{equation*}
 \SEP(\W_N) = \Noses_1(\Len+\Thres)  + (\Noses_{-\Height} + \Noses_{1-\Height})(\Len+\Thres-\Circ) + (\Thres+\BegDist)(\Steps_0 + \Steps_1) + (\Thres+\BegDist-\Circ)\Steps_{-\Height}
\end{equation*}
while by \eqref{eq:walk jump}
\begin{equation*}
 \Noses_1 = -\Steps_1 + \Height(\Noses_{-\Height} + \Steps_{-\Height}) + (\Height-1)\Noses_{1-\Height}
\end{equation*}
thus
\begin{equation}
 \Circ \geq (\Len+\Thres)\left(\Height + \frac{\Noses_{-\Height} - \Steps_1}{\Noses_{1-\Height} + \Noses_{-\Height} + \Steps_{-\Height}}\right). \label{eq:nose ratio}
\end{equation}
For any nose walk $\W_N$, the value $\Ratio(\W_N) = \frac{\Noses_{-\Height} - \Steps_1}{\Noses_{1-\Height} + \Noses_{-\Height} + \Steps_{-\Height}}$ is referred to as the \emph{ratio} of $\W_N$, while the \emph{nose ratio} of $\M$ is $\Ratio(\M) = \max\{\Ratio(\W_N) \mid \W_N \text{ is a nose cycle of } \M\}$.  

A similar analysis is enough to conclude (assuming $x / 0 = \infty$) that
\begin{equation}
 \Circ \leq (\Len-\Thres)\left(\Height + \frac{\Hollows_0 + \Hollows_\Height + \Steps_1}{\Hollows_\Height + \Hollows_{\Height-1} - \Steps_{-\Height}}\right) \label{eq:hollow ratio}
\end{equation}
for any hollow cycle $\W_H$ (this is the reason why hollow cycles are restricted to $\Hollows_h + \Hollows_{h-1} \geq \Steps_{-h}$ by definition).  This time, for any hollow walk $\W_H$, the value $\RATIO(\W_H) = \frac{\Hollows_0 + \Hollows_\Height + \Steps_1}{\Hollows_\Height + \Hollows_{\Height-1} - \Steps_{-\Height}}$ is said to be the \emph{ratio} of $\W_H$, while $\RATIO(\M) = \min\{\RATIO(\W_H) \mid \W_H \text{ is a hollow cycle of } \M\}$ is the \emph{hollow ratio} of $\M$.   The following observation sums up equations \eqref{eq:nose ratio}~and~\eqref{eq:hollow ratio}; note that, as usual, we omit the parameter $\M$ from $\Ratio$ and $\RATIO$.

\begin{observation}
 For every $\Descriptor$-CA model, 
 \begin{equation}
    (\Len+\Thres)(\Height + \Ratio) \leq \Circ \leq (\Len-\Thres)(\Height + \RATIO) \label{eq:ratios}
 \end{equation}
\end{observation}

By~\eqref{eq:ratios}, $\M$ is equivalent to a $\Descriptor$-CA model only if $\Circ = (\Len+\Thres)(\Height+\Ratio) + \Extra$ for some $\Extra \geq 0$.  The factor $(\Len+\Thres)(\Height + \Ratio)$ is required for each nose cycle to fit in the model when considered in isolation, while the \emph{extra space} \Extra\ serves to accommodate the interactions between all the arcs.  Note that, in general, $\M$ needs not be equivalent to a $\Descriptor$-CA model.  This is not important, though, as we can always write $\Circ$ as $(\Len+\Thres)(\Height+\Ratio) + \Extra$; just observe that $\Extra$ could be negative in some cases.  

We find it convenient to express $\SEP(\W)$ as a function of $\Len$ and $\Extra$, for every walk $\W$ of $\Syn$.  With this in mind, observe that the value of $\SEP(\W)$ for a walk $\W$ is, by definition
\begin{align*}
&\Noses_1(\Len+\Thres) + (\Noses_{-\Height} + \Noses_{1-\Height})(\Len+\Thres-\Circ) + \phantom{x}\\& (\Hollows_0 + \Hollows_{-1})(\Thres-\Len) + (\Hollows_{\Height} + \Hollows_{\Height-1})(\Circ+\Thres-\Len) + \phantom{x}\\& (\Thres+\BegDist)(\Steps_0 + \Steps_1) + (\Thres+\BegDist-\Circ)\Steps_{-\Height}.
\end{align*}
Applying Equation~\eqref{eq:walk jump} and some algebraic manipulation, we conclude that 
\begin{equation}
\SEP(\W) = (\Len+\Thres)(\Height(A_j) - \Height(A_i) + \MAG(\W)) + e\REM(\W) + \CONST(\W)\label{eq:sep}
\end{equation}
where:
\begin{align}
  \MAG(\W)   &= \Noses_{-\Height} - \Hollows_\Height -\Steps_{1} - \Hollows_0 + \Ratio\REM(\W) \label{eq:mag}\\
  \REM(\W)   &= \Hollows_\Height + \Hollows_{\Height-1} - \Noses_{-\Height} - \Noses_{1-\Height} - \Steps_{-\Height} \label{eq:rem}\\
  \CONST(\W, \Thres, \BegDist) &= 2\Thres(\Hollows_{-1} + \Hollows_{0} + \Hollows_{\Height} + \Hollows_{\Height-1}) + (\Thres+\BegDist)(\Steps_1 + \Steps_0 + \Steps_{-\Height})\label{eq:const}
\end{align}
 
The values $\MAG(\W)$, $\REM(\W)$, and $\CONST(\W, \Thres, \BegDist)$ are the \emph{length}, \emph{extra}, and \emph{constant factors} of $\W$, and $\W$, $\Thres$, and $\BegDist$ are omitted as usual.  We emphasize that the triplet $(\MAG, \REM, \CONST)$ is, in some sort, a generalization to what Mitas takes as the column in his pictorial representation (see~\cite{Mitas1994} and Section~\ref{sec:minimal uig}).  The main difference is that the external edges can be disregarded from $\Syn$ when $\M$ is a PIG model.  Mitas also discards the $0$-hollows and the steps to define the column, thus $(\MAG, \REM, \CONST)$ gets reduced to $(0,0, 2\Hollows_{-1})$.

There are at least two advantages of expressing $\SEP$ as a polynomial with indeterminates $\Len$ and $\Extra$ and coefficients $\MAG, \REM$ and $\CONST$.  First, by Theorem~\ref{thm:separation constraints}, we can see at first sight that $\M$ is equivalent to a $(\Circ, \Len)$-CA model whenever either $\MAG(\W) < 0$ or $\MAG(\W) = 0$ and $\REM(\W) < 0$ for every cycle $\W$.  Just take large enough values for $\Len$ and $\Extra$.  In particular, observe that $\MAG(\W) < 0$ for every cycle when $\M$ is a PIG model; thus, this is just one more proof of the fact that every PIG model is equivalent to an UIG model.  The second advantage is that, obviously, the factors depend only on the structure of $\Syn$ and not on the weighing function $\SEP$.  In fact, we can compute $\REM(\W)$ by means of the edge weighing $\REM$ (the overloaded notation is intentional) of $\Syn$ such that 
\[
 \REM(A_i \to A_j) =
  \begin{cases}
   1 & \text{if $A_i \to A_j$ is an $h$-hollow or $(h-1)$-hollow} \\
   -1 & \text{if $A_i \to A_j$ is a ($-h$)-nose, a $(1-h)$-nose, or a $(-h)$-step} \\
   0 & \text{otherwise}
  \end{cases}
\]
We can compute $\MAG(\W)$ and $\CONST(\W, \Thres, \BegDist)$ in a similar fashion with the corresponding edge weighings $\MAG$ and $\CONST_{\Thres, \BegDist}$.

\section{Efficient Tucker's characterization}
\label{sec:tucker}

In this section we give an alternative proof of Tucker's characterization, taking advantage of the framework of synthetic graphs.  In short, Tucker's theorem states that $\M$ is equivalent to some UCA model if and only if $a/b < x/y$ for every $(a,b)$-independent and every $(x,y)$-circuit of $\M$~\cite{TuckerDM1974}.  As already mentioned (and proven in Section~\ref{sec:circuits and independents}) the nose and hollow cycles of $\Syn$ are the equivalents of the $(a,b)$-independents and $(x,y)$-circuits of $\M$.  Moreover, the maximal and minimal values of $a/b$ and $x/y$ are somehow related to $\Ratio$ and $\RATIO$, respectively.  Thus, intuition tells us that we should be able to prove that $\M$ is equivalent to a UCA model if and only if $\Ratio < \RATIO$.  This is equivalence (\ref{thm:tucker:equivalence}) $\Leftrightarrow$ (\ref{thm:tucker:bound}) of Theorem~\ref{thm:tucker} below.

Though equivalence (\ref{thm:tucker:equivalence}) $\Leftrightarrow$ (\ref{thm:tucker:bound}) is not new, our proof of this fact is new and somehow simple.  One of the main features about Theorem~\ref{thm:tucker} is that it exhibits new characterizations that can be used for positive and negative certification.  In particular, it shows how to obtain an integer $(\Circ, \Len)$-CA model equivalent to $\M$ with $\Circ$ and $\Len$ polynomial in $n$.  The existence of such models was questioned by Durán et al.\ in~\cite{DuranGravanoMcConnellSpinradTuckerJA2006} and proved by Lin and Szwarcfiter in~\cite{LinSzwarcfiterSJDM2008} by means of feasible circulations. 

Before stating Theorem~\ref{thm:tucker}, we study the relation between $\SEP$ and the ratios of $\M$.  Recall that the $\SEP$-values of nose and hollow cycles impose the lower and upper bounds described by \eqref{eq:ratios}, respectively.  The reason to consider only nose and hollow cycles is that they have the largest $\SEP$-values when $\Circ$ and $\Len$ are large, as it follows from \eqref{eq:sep} and the next lemma.

\begin{lemma}\label{lem:nose or hollow cycle}
 For any walk $\W$ of $\Syn$ there exists either a nose or hollow walk $\W'$ of $\Syn$ starting and ending at the same vertices as $\W$ such that $\MAG(\W) \leq \MAG(\W')$ and $\REM(\W) \leq \REM(\W')$.
\end{lemma}

\begin{proof}
 The proof is by induction on $\Noses(\W) \times \Hollows(\W)$ and $\Steps(\W) + 1$, the base case of which is trivial.  Suppose, then, that $\W$ has at least one nose and one hollow.  So, $\W$ must have a subwalk $\W_{1i} = B_1, \ldots, B_i$ such that $B_1 \to B_2$ is a nose, $B_2, \ldots, B_{i-1}$ is a step walk, and $B_{i-1} \to B_i$ is a hollow.  Observe that $i \geq 3$, because $t(B_1)s(B_2)$ are consecutive and thus $B_2 \to B_1$ is not a hollow.
 
 Consider first the case in which $\W_{1i}$ is not a path, thus it contains a cycle $\W_{jk}$ $=$ $B_j$, $\ldots$, $B_k = B_j$ ($j < k$) that must have at least one $1$-step or $0$-hollow.  Note that $\W_{jk} \neq \W_{1i}$ because otherwise $B_2, \ldots, B_{i-1}$ would pass through $B_1 = B_i$ contradicting the fact that $\W_{jk}$ is a cycle.  Hence, $\W_{jk}$ does not contains both a nose and a hollow and $\MAG(\W_{jk}) \leq 0$; recall \eqref{eq:mag}.  Moreover, if $\W_{jk}$ has an external hollow (which must be $B_{i-1} \to B_i$), then it must contain the unique external step of $\Syn$.  Therefore, $\REM(\W_{ji}) \leq 0$ by \eqref{eq:rem}, and the proof follows by induction on $\W \setminus \W_{ji}$.  

 Consider now the case in which $\W_{1i}$ is a path and let $\W_{1i}'$ be the step path from $B_1$ to $B_i$.  We claim that $\MAG(\W_{1i}) = \MAG(\W_{1i}')$ and $\REM(\W_{1i}) = \REM(\W'_{1i})$, in which case the proof follows by induction on $(\W \setminus \W_{1i}) \cup \W'_{1i}$.  Since $\W_{1i}$ is a path, it follows that either $B_2 > B_i$ or $B_i > B_1$, which leaves us with only five possible combinations for the heights of $B_1$, $B_2$, $B_{i-1}$, and $B_{i}$, all of which are analyzed in the table below.  The claim is therefore true.

 \begin{center}
  \begin{tabular}{|c|c|c|c|c|c|c|c|}
    \hline
    $\Height(B_1)$ & $\Height(B_2)$ & $\Height(B_{i-1})$ & $\Height(B_{i})$ & $\MAG(\W_{1i})$ & $\MAG(\W'_{1i})$ & $\REM(\W_{1i})$ & $\REM(\W_{1i})$ \\\hline
    $x$ & $y$ & $y$ & $x$ & $0$ & $0$ & $0$ & $0$ \\
    $\Height$ & $0$ & $0$ or $1$ & $0$       & $-\Ratio$ & $-\Ratio$ & $-1$ & $-1$ \\
    $\Height - 1$ & $0$ & $0$ or $1$ & $0$   & $-1 - \Ratio$ & $-1 - \Ratio$ & $-1$ & $-1$ \\
    $\Height - 1$ & $\Height$ or $0$ & $0$ & $\Height$    & $-1$ & $-1$ & $0$ & $0$ \\
    $\Height - 1$ & $\Height$ & $0$ & $\Height-1$    & $0$ & $0$ & $0$ & $0$ \\\hline
  \end{tabular}
 \end{center}
\end{proof}

The above lemma brings us closer to Tucker's characterization, as it shows that any cycle with $\SEP > 0$ can be transformed into a hollow or nose cycle; thus, the existence of a UCA model equivalent to $\M$ is reduced to how its ratios look like.  One of the salient features of our proof is that it builds an efficient UCA model $\Unit$ equivalent to $\M$.  The idea is to take $\Unit = \Unit(\M, \Circ, \Len)$ as in Theorem~\ref{thm:separation constraints} for some appropriate values of $\Circ$ and $\Len$.  Observe that $\LeftBound = \RightBound = 0$ in this case, hence we can replace $\BoundedSyn$ and $A_0$ with $\Syn$ and $A_1$ in the definition of $\Unit$.  In principle, $O(n^2)$ time is required to compute all the values of $\IDist{\SEP}(A_1, A_i)$, because $\Syn$ is not an acyclic graph.  By taking some appropriate values for $\Circ$ (or $\Extra$) and $\Len$, we can remove all these cycles so as to reduce the time complexity to $O(n)$.  With this in mind, we say that an edge $A_i \to A_j$ of $\Syn$ is \emph{redundant} when either 
\begin{enumerate}[($\rm red_1$)]
  \item $\IDist{\MAG}(A_1, A_j) > \IDist{\MAG}(A_1, A_i) + \MAG(A_i \to A_j)$, or  \label{def:red-1}
  \item $\IDist{\MAG}(A_1, A_j) = \IDist{\MAG}(A_1, A_i) + \MAG(A_i \to A_j)$ and\\ 
        $(\IDist{\REM} \circ \IDist{\MAG})(A_1, A_j) > (\IDist{\REM} \circ \IDist{\MAG})(A_1, A_i) + \MAG(A_i \to A_j)$. \label{def:red-2}
\end{enumerate}
\newcommand{\redref}[1]{\itemref{red}{red}{#1}}
Roughly speaking, $A_i \to A_j$ is redundant when it plays no role on the separation between $s(A_1) = 0$ and $s(A_j)$ for large values of $\Len$ and not-so-large values of $\Extra$.  (Recall that $(\IDist{\REM} \circ \IDist{\MAG})$ is the $\REM$-distance restricted only to those paths with maximum $\MAG$-distance.) The reduction of $\Syn(\M)$ is the digraph $\Red(\M)$ obtained after removing all the redundant edges of $\Syn(\M)$; as usual, we omit the parameter $\M$.  Theorem~\ref{thm:tucker} includes Tucker's characterization as equivalence (\ref{thm:tucker:equivalence}) $\Leftrightarrow$ (\ref{thm:tucker:bound}).

\begin{theorem}\label{thm:tucker}
  Let $\M$ be a PCA model with arcs $A_1 < \ldots < A_n$.  Then, the following statements are equivalent:

  \begin{enumerate}[(i)]
    \item $\M$ is equivalent to a UCA model.\label{thm:tucker:equivalence}
    \item $\Ratio < \RATIO$.\label{thm:tucker:bound}
    \item $\MAG(\W) < 0$ for every hollow cycle $\W$ of $\Syn$.\label{thm:tucker:hollow-cycles}
    \item either $\MAG(\W) < 0$ or $\MAG(\W) = 0$ and $\REM(\W) < 0$, for every cycle $\W$ of $\Syn$.\label{thm:tucker:cycles}
    \item $\Red$ is acyclic.\label{thm:tucker:negative-certificate}
    \item $\Dist{\SEP_{(\Circ, \Len)}}(\Syn,A_1, A_i) = \IDist{\SEP_{(\Circ, \Len)}}(\Red, A_1, A_i)$ for every $1 \leq i \leq n$, where $\Circ = (\Len+1)(\Height+\Ratio) + \Extra$, $(\Len+1) = \Ratio_2\Extra^2$, $\Extra = 4n$ and $\Ratio = \Ratio_1/\Ratio_2$ for $\Ratio_1, \Ratio_2 \in\mathbb{N}$.\label{thm:tucker:positive-certificate}
    \item $\Unit(\M, \Circ, \Len)$ is an integer $(\Circ, \Len)$-CA model equivalent to $\M$ for $\Circ$ and $\Len$ as in~(\ref{thm:tucker:positive-certificate}).\label{thm:tucker:model}
  \end{enumerate}
\end{theorem}

\begin{proof}
  (\ref{thm:tucker:equivalence}) $\Rightarrow$ (\ref{thm:tucker:bound}).  This is direct consequence of \eqref{eq:ratios}.

  (\ref{thm:tucker:bound}) $\Rightarrow$ (\ref{thm:tucker:hollow-cycles}).  If $\W_H$ is a hollow cycle with a nonnegative length factor, then
    \[0 \leq -\Steps_1 - \Hollows_0 - \Hollows_{\Height} + \Ratio(\Hollows_{\Height} + \Hollows_{\Height - 1} - \Steps_{-\Height})\]
    implying (recall \eqref{eq:hollow ratio} observing that $\Hollows_{\Height} + \Hollows_{\Height - 1} \geq \Steps_{-\Height}$)
    \[\Ratio \geq \frac{\Hollows_\Height+\Hollows_0+\Steps_1}{\Hollows_\Height + \Hollows_{\Height-1} - \Steps_{-\Height}} = \RATIO(\W_H) \geq \RATIO.\]
  
  (\ref{thm:tucker:hollow-cycles}) $\Rightarrow$ (\ref{thm:tucker:cycles}).  Suppose either $\MAG(\W) > 0$ or $\MAG(\W) = 0$ and $\REM(\W) \geq 0$ for some cycle $\W$ of $\Syn$.  If $\REM(\W) \geq 0$, then the statement follows as there is a hollow cycle with a nonnegative length factor by Lemma~\ref{lem:nose or hollow cycle}.  Otherwise, there is a nose cycle $\W_N$ with positive length factor by Lemma~\ref{lem:nose or hollow cycle}, so
    \[0 < \Noses_{-\Height} -\Steps_1 - \Ratio(\Noses_{-\Height} + \Noses_{1-\Height} + \Steps_r)\]
    implying (recall \eqref{eq:nose ratio})
    \[\Ratio < \frac{\Noses_{-\Height} - \Steps_1}{\Noses_{-\Height} + \Noses_{1-\Height} + \Steps_{-\Height}} = \Ratio(\W_N) \leq \Ratio,\]
    which is impossible.

  (\ref{thm:tucker:cycles}) $\Rightarrow$ (\ref{thm:tucker:negative-certificate}).  Suppose $\Red$ has some cycle $\W = B_1, \ldots, B_k$ with $B_1 = B_{k}$.  By \redref{1}, 
  \begin{equation*}
    \IDist{\MAG}(A_1, B_{i+1}) \leq \IDist{\MAG}(A_1, B_i) + \MAG(B_i \to B_{i+1}).  \tag{a} 
  \end{equation*}
  Then, by induction, 
  \[\IDist{\MAG}(A_1, B_1) = \IDist{\MAG}(A_1, B_{k}) \leq \IDist{\MAG}(A_1, B_1) + \MAG(\W),\]
  which implies that $\MAG(\W) \geq 0$.  Moreover, $\MAG(\W) = 0$ only if (a) holds by equality for every $1 \leq i \leq k$, thus 
  \[(\IDist{\REM} \circ \IDist{\MAG})(A_1, B_{i+1}) \leq (\IDist{\REM} \circ \IDist{\MAG})(A_1, B_i) + \REM(B_i \to B_{i+1})\] 
  by \redref{2}, implying $\REM(\W) \geq 0$ by induction.

  (\ref{thm:tucker:negative-certificate}) $\Rightarrow$ (\ref{thm:tucker:positive-certificate}).  Taking into account that $\Red$ is acyclic and every walk of $\Red$ is also a walk of $\Syn$, it follows that $\Dist{\SEP}(\Syn, A_1, A_i) \geq \IDist{\SEP}(\Red, A_1, A_i)$ for every $1 \leq i \leq n$.  
  
  For the remaining inequality suppose, by induction, that $\SEP(\W) \leq \IDist{\SEP}(\Red, A_1, A_i)$ for every walk $\W$ of $\Syn$ that goes from $A_1$ to $A_i$ whose length is at most $k-1 < n$.  Consider any walk $\W = B_1, \ldots, B_{k+1}$ of $\Syn$ from $A_1 = B_1$, and let
  \begin{itemize}
    \item $\W^\Red_j$ be a walk of $\Red$ from $B_1$ to $B_j$ with $\SEP(W^\Red_j) = \IDist{\SEP}(\Red, B_1, B_j)$, and
    \item $\W^\Syn_{k+1}$ be the walk obtained by traversing $B_{k} \to B_{k+1}$ after $\W^\Red_{k}$.
  \end{itemize}
  By inductive hypothesis, $\SEP(\W) \leq \SEP(\W^\Red_{k}) + \SEP(B_{k} \to B_{k+1}) = \SEP(\W^\Syn_{k+1})$, thus $\SEP(\W) \leq \SEP(\W^\Red_{k+1})$ when $B_{k} \to B_{k+1}$ is an edge of $\Red$.  Suppose, then, that $B_{k} \to B_{k+1}$ is redundant in $\Syn$, and consider the two possibilities according to \redref{1} and \redref{2}.
  \begin{description}
    \item [Case 1:] \redref{1} is true.  Note that since no edge of $\W_{k+1}^\Red$ is redundant, it follows by induction that $\MAG(\W_j^\Red) = \IDist{\MAG}(A_1, B_j)$ for every $1 \leq j \leq k+1$.  Hence, 
    \[\MAG(\W_{k+1}^\Red) = \IDist{\MAG}(A_1, B_{k+1}) > \IDist{\MAG}(A_1, B_{k}) + \MAG(B_{k} \to B_{k+1}) = \MAG(\W_{k+1}^\Syn).\]
    Now, taking into account that every term of the length factor is a multiple of either $1$ or $\Ratio = \Ratio_1/\Ratio_2$ in \eqref{eq:mag}, we obtain that 
    \[(\Len+1)(\MAG(\W_{k+1}^\Red) - \MAG(\W^\Syn_{k+1})) = \Ratio_2\Extra^2(\MAG(\W^\Red_{k+1}) - \MAG(\W^\Syn_{k+1})) \geq \Extra^2.\]  By \eqref{eq:sep}, we obtain that
    \begin{align*}
      \SEP(\W_{k+1}^\Red) \geq& (\Len+1)(\Height(B_{k+1}) + \MAG(\W_{k+1}^\Red)) - \Extra n \\
      \SEP(\W^\Syn_{k+1}) \leq& (\Len+1)(\Height(B_{k+1}) + \MAG(\W_{k+1}^\Syn)) + \Extra k + 2k,
    \end{align*}
    thus,
    \[\SEP(\W_{k+1}^\Red) - \SEP(\W^\Syn_{k+1}) \geq \Extra^2 - \Extra(k+n) - 2k \geq 8n - 2k \geq 2.\]
    \item [Case 2:] \redref{1} is false, thus \redref{2} holds.  As before, we observe by induction that $\MAG(\W_j^\Red) = \IDist{\MAG}(A_1, B_j)$ and, thus, $\REM(\W_j^\Red) = (\IDist{\REM} \circ \IDist{\MAG})(A_1, B_j)$.  Consequently, by \redref{2},
    \[\REM(\W_{k+1}^\Red) > (\IDist{\REM} \circ \IDist{\MAG})(A_1, B_{k}) + \REM(B_{k} \to B_{k+1}) = \REM(\W_{k+1}^\Syn).\]
    Since \redref{1} is true, it follows that $\MAG(\W_{k+1}^\Syn) = \IDist{\MAG}(A_1, B_{k+1})$.  Then, by \eqref{eq:sep}, 
    \[\SEP(\W_{k+1}^\Red) - \SEP(\W_{k+1}^\Syn) \geq \Extra(\REM(\W_{k+1}^\Red) - \REM(\W^\Syn_{k+1})) - \CONST(\W^\Syn_{k+1}) \geq 4n - 2k \geq 2.\]
  \end{description}
  We conclude, therefore, that $\SEP(\W) \leq \IDist{\SEP}(\Red, A_1, A_i)$ for every walk $\W$ of $\Syn$ that goes from $A_1$ to $A_i$ whose length is at most $n$.  By induction, this implies that $\SEP(\W) \leq \IDist{\SEP}(\Red, A_1, A_i)$ whatever the length of $\W$ is, thus $\Dist{\SEP}(\Syn, A_1, A_i) = \IDist{\SEP}(\Red, A_1, A_i)$.
  
  (\ref{thm:tucker:positive-certificate}) $\Rightarrow$ (\ref{thm:tucker:model}) Since $\Dist{\SEP}(\Syn, A_1, A_i) = \IDist{\SEP}(\Red, A_1, A_i)$ for every $1 \leq i \leq n$ and $\Syn$ is strongly connected, it follows that $\Dist{\SEP}(\Syn, A_i, A_j) = \IDist{\SEP}(\Syn, A_i, A_j)$ for every $1 \leq i,j \leq n$.  Hence $\SEP(\W) \leq 0$ for every cycle $\W$ of $\Syn$ and the implication follows by Theorem~\ref{thm:separation constraints} (note that $\Circ$ and $\Len$ are integer values).

  (\ref{thm:tucker:model}) $\Rightarrow$ (\ref{thm:tucker:equivalence}). Trivial.
\end{proof}

Theorem~\ref{thm:tucker} has some nice algorithmic consequences \Rep when combined with Theorem~\ref{thm:separation constraints}.  For any input PCA model $\M$ we solve \uRep for the UCA descriptor $\Descriptor$ implied by statement~(\ref{thm:tucker:positive-certificate}).  As a byproduct, we either obtain a UCA model equivalent to $\M$ or a cycle of $\Syn$ that can be used for negative certification.  The algorithm costs $O(n^2)$ time, plus the time and space required so as to compute $\Ratio(\M)$.  In Section~\ref{sec:linear algorithm} we show a not-so-hard $O(n)$ time variation of this algorithm, taking advantage of the reduction of $\Syn$.  However, we first discuss how $\Ratio$ can be found.

\section{The recognition algorithm by Kaplan and Nussbaum}
\label{sec:kaplan and nussbaum}

Translated to synthetic graphs, Tucker's characterization (equivalence (\ref{thm:tucker:equivalence}) $\Leftrightarrow$ (\ref{thm:tucker:bound}) of Theorem~\ref{thm:tucker}) states that $\M$ is equivalent to no UCA model if only if $\Syn$ has nose and hollow cycles $\W_N$ and $\W_H$ such that $\Ratio(W_N) \geq \RATIO(W_H)$. The original proof by Tucker does not show how to obtain such cycles.  More than thirty years later, in~\cite{DuranGravanoMcConnellSpinradTuckerJA2006}, Durán et al.\ described the first polynomial algorithm to obtain such cycles with a rather complex implementation.  A few years later, in~\cite{KaplanNussbaumDAM2009}, Kaplan and Nussbaum improved this algorithm so as to run in $O(n)$ time while simplifying the implementation.  The purpose of this section is to translate the algorithm by Kaplan and Nussbaum in terms of the synthetic graph.  The proof of correctness is simple, short, and rather intuitive, while the implementation of the algorithm is quite similar to the one given by Kaplan and Nussbaum.

The main concept of this section is that of greedy cycles.  For any nose (resp.\ hollow) walk $\W_N = B_1, \ldots, B_k$ of $\Syn$, we say that $B_i$ is \emph{greedy (in $W_N$)} when either $B_i \to B_{i+1}$ is nose (resp.\ hollow) or no nose (resp.\ hollow) of $\Syn$ starts at $B_i$.  A nose (resp.\ hollow) cycle is \emph{greedy} when all its vertices are greedy.  In other words, $W_N$ is greedy when noses (resp.\ hollow) are preferred over steps.  The main idea of Durán et al., which was somehow implicit in~\cite{TuckerDM1974}, is to observe that $\Syn$ contains a greedy nose (resp.\ hollow) cycle of highest (resp.\ lowest) ratio.  Then, they compute the unique greedy nose (resp.\ hollow) cycle starting at a vertex $A$, for every $A \in V(\Syn)$, and keep the one with highest (resp.\ lowest) ratio.  Note that each greedy nose (resp.\ hollow) cycle $B_1, \ldots, B_k$ is found $k$ times, once for each starting vertex $B_i$.  Kaplan and Nussbaum, instead, compute each greed nose (resp.\ hollow) cycle only once by taking only one vertex as the starting point.  

The next lemma has a new proof that $\Syn$ contains a greedy nose cycle of highest ratio.

\begin{lemma}[see also~\cite{DuranGravanoMcConnellSpinradTuckerJA2006,KaplanNussbaumDAM2009,TuckerDM1974}]\label{lem:greedy nose cycle}
 For any nose cycle $\W_N$ of $\Syn$ there exists a greedy nose cycle $\W_N'$ of $\Syn$ such that $\Ratio(\W_N) \leq \Ratio(\W_N')$.
\end{lemma}

\begin{proof}
 The proof is trivial when $\W_N$ is greedy.  When $\W_N$ is not greedy, we can transform it into a greedy nose cycle by traversing $\W_N$ from any vertex while applying the following operation when a non-greedy vertex $B_1$ is found, until no more non-greedy vertices remain.  Let $B_1 \to B$ be the nose from $B_1$ and $\W = B_1, \ldots, B_i$ be the shortest subpath of $\W_N$ such that either $B_i = B$ or $B_{i-1} \to B_{i}$ is a nose.  The operation transforms $\W_N$ into $\W_N'$ by replacing $\W$ with $\W'$, where $\W'$ is the path formed by the nose $B_1 \to  B$ followed by step path from $B$ to $B_i$.  Thus, it suffices to prove that $\Ratio(\W) \leq \Ratio(\W')$.  Moreover, if we write $\Ratio_1$ and $\Ratio_2$ to denote the numerator and denominator of $\Ratio$ as in~\eqref{eq:nose ratio}, then $\Ratio_j(\W_N') = \Ratio_j(\W_N) - \Ratio_j(\W) + \Ratio_j(\W')$ for $j \in \{1,2\}$.  So, it is enough to show that $\Ratio_1(\W') \geq \Ratio_1(\W)$ and $\Ratio_2(\W') \leq \Ratio_2(\W)$.
 
 If $B = B_i$ then $\W$ has at least one $1$-step or $(-h)$-step, while $\W' = B_1, B$.  Thus $\Ratio_1(\W) \leq 0 \leq \Ratio_1(\W')$ and $\Ratio_2(\W') \leq \Ratio_2(\W)$, and the lemma follows.  Otherwise, if $B \not\in \W$, then $B_1, B_{i-1}, B, B_{i}$ appear in this order in $\Syn$ when its steps are traversed from $B_1$.  It is not hard to see that $\Ratio_i(\W) = \Ratio_i(\W')$ when $B_1 \to B$ and $B_{i-1} \to B_{i}$ have equal jumps.  This leaves us with only four possible combinations for the heights of $B_1$, $B_{i-1}$, $B_i$, and $B$ when the jumps differ, and in all such cases the lemma is true (see the table below).
 
 \begin{center}
  \begin{tabular}{|c|c|c|c|c|c|c|c|}
    \hline
    $\Height(B_1)$ & $\Height(B_{i-1})$ & $\Height(B)$ & $\Height(B_{i})$ & $\Ratio_1(\W)$ & $\Ratio_1(\W')$ & $\Ratio_2(\W)$ & $\Ratio_2(W')$ \\\hline
    $\Height-2$ & $\Height-1$ & $\Height-1$        & $0$        & $-1$ & $-1$ & $1$ & $1$ \\
    $\Height-1$ & $\Height-1$ & $\Height$          & $0$        & $0$  & $0$  & $1$ & $1$ \\
    $\Height-1$ & $\Height$   & $\Height$ or $0$   & $0$        & $0$  & $0$  & $1$ & $1$ \\
    $\Height$   & $0$         & $0$                & $1$        & $0$ & $0$  & $1$  & $1$ \\\hline
  \end{tabular}
 \end{center}
\end{proof}

The proof that $\Syn$ contains a greedy hollow cycle with lowest ratio is similar and we omit it as it not required by our algorithm.  Moreover, an analogous proof is given in Lemma~\ref{lem:greedy walks}, while Section~\ref{sec:circuits and independents} shows the equivalence between hollow cycles and $(x,y)$-circuits.

\begin{lemma}[see~\cite{KaplanNussbaumDAM2009} and Section~\ref{sec:circuits and independents}]\label{lem:greedy hollow cycle}
 For any hollow cycle $\W_H$ of $\Syn$ there exists a greedy hollow cycle $\W_H'$ of $\Syn$ such that $\RATIO(\W_H) \geq \RATIO (\W_H')$.
\end{lemma}

The algorithm to compute a nose (resp.\ hollow) cycle with highest (resp.\ lowest) ratio follows easily from Lemma~\ref{lem:greedy nose cycle} (resp.\ Lemma~\ref{lem:greedy hollow cycle}).  Just note that if an edge $A_i \to A_j$ belongs to a greedy nose (resp.\ hollow) cycle, then either $A_i \to A_j$ is a nose (resp.\ hollow), or there are no noses (resp.\ hollow) from $A_i$ in $\Syn$.  Then, $\W$ is a greedy nose (resp.\ hollow) cycle of $\Syn$ if and only if $\W$ is a cycle of the digraph $\Syn_N$ that is obtained by keeping only the noses (resp.\ hollows) of $\Syn$ and the steps that go from vertices with no noses (resp.\ hollows).  Since all the vertices in $\Syn_N$ have out-degree $1$, we can obtain all the greedy cycles in $O(n)$ time.  Then, by Lemmas \ref{lem:greedy nose cycle}~and~\ref{lem:greedy hollow cycle}, $\Ratio$ and $\RATIO$ can be computed in $O(n)$ time.  Furthermore, if $\Ratio > \RATIO$, then a nose and a hollow cycles $\W_N$ and $\W_H$ with $\Ratio(\W_N) = \Ratio$ and $\RATIO(\W_H) = \RATIO$ are obtained as a byproduct.

\section{Efficient construction of UCA models}
\label{sec:linear algorithm}

In~\cite{DuranGravanoMcConnellSpinradTuckerJA2006}, Durán et al.\ ask if there exists an integer $(\Circ, \Len)$-CA model equivalent to a PCA model $\M$ such that $\Circ$ and $\Len$ are bounded by a polynomial in $n$.  If affirmative, they also inquire whether such a model can be found in $O(n)$ time.  Both questions were affirmatively answered by Lin and Szwarcfiter in~\cite{LinSzwarcfiterSJDM2008}, who showed how to reduce the problem of finding such a $(\Circ, \Len)$-CA model to a circulation problem.  Their algorithm and its correctness have nothing to do with Theorem~\ref{thm:tucker} and it is not easy to see how a forbidden subgraph can be obtained for certification (that is, without invoking a second recognition algorithm, as the one by Kaplan and Nussbaum).  Kaplan and Nussbaum ask for a \emph{unified} certification algorithm.  We provide such an algorithm in this section.

Our algorithm is based on the equivalence (\ref{thm:tucker:equivalence}) $\Leftrightarrow$ (\ref{thm:tucker:negative-certificate}) $\Leftrightarrow$ (\ref{thm:tucker:positive-certificate}) of Theorem~\ref{thm:tucker}.  That is, the algorithm just checks whether $\Red(\M)$ is acyclic.  If affirmative, then $\Unit(\Circ, \Len)$ can be taken as the positive certificate by Theorem~\ref{thm:separation constraints}, where $\Circ$ and $\Len$ are defined as in Theorem~\ref{thm:tucker}(\ref{thm:tucker:positive-certificate}).  Otherwise, a nose cycle with ratio $\Ratio(\M)$ combined with a cycle of $\Red(\M)$ form the negative certificate.  Of course, testing if $\Red$ is acyclic and finding a cycle in $\Red$ both cost $O(n)$ time.  When $\Red$ is acyclic, we can compute $\IDist{\SEP}$ in $O(n)$ time using $\Red$ so as to build $\Unit(\Circ, \Len)$.  Hence, the difficulty of the algorithm is on finding $\Red$.  

By definition, $\Red$ is obtained by removing the redundant edges of $\Syn$; this can be done in $O(n)$ time once the redundant edges of $\Syn$ are found.  In turn, recall that an edge $A_i \to A_j$ is redundant if and only if 
\begin{enumerate}[($\rm red_1$)]
  \item[\redref{1}] $\IDist{\MAG}(A_1, A_j) > \IDist{\MAG}(A_1, A_i) + \MAG(A_i \to A_j)$, or
  \item[\redref{2}] $\IDist{\MAG}(A_1, A_j) > \IDist{\MAG}(A_1, A_i) + \MAG(A_i \to A_j)$ and\\ 
                    $(\IDist{\MAG} \circ \IDist{\REM})(A_1, A_j) > (\IDist{\MAG} \circ \IDist{\REM})(A_1, A_i) + \REM(A_i \to A_j)$.
\end{enumerate}
Thus, to locate the redundant edges we should find a path $\W_i$ from $A_1$ to $A_i$ such that $\MAG(\W_i) = \IDist{\MAG}(A_1, A_i)$ and $\REM(\W_i) = (\IDist{\MAG} \circ \IDist{\REM})(A_1, A_i)$ for every $1 \leq i \leq n$.  By Lemma~\ref{lem:nose or hollow cycle}, $\W_i$ is either a nose or hollow cycle.  That is, $\W_i \in \{\W_{N,i}, \W_{H,i}\}$, where $\W_{N,i}$ is the nose cycle such that
\begin{enumerate}[(gn$_1$)]
  \item $\MAG(\W_{N, i}) = \max\{\MAG(\W) \mid \W \text{ is a nose path from } A_1 \text{ to } A_i\}$, and \label{def:gn-1} 
  \item $\REM(\W_{N,i}) = \max\{\REM(\W) \mid \W \text{ is a nose path from } A_1 \text{ to } A_i \text{ with } \MAG(\W) = \MAG(\W_{N,i})\}$, \label{def:gn-2}
\end{enumerate}
\newcommand{\gnref}[1]{\itemref{gn}{gn}{#1}}
while $\W_{H,i}$ is defined analogously by replacing noses with hollows.  The remainder of this section is devoted to the problems of finding $\W_{N,i}$ and $\W_{H,i}$.

Say that a nose (resp.\ hollow) walk $\W = B_1, \ldots, B_k$ is \emph{greedy} when there exists $B_i \in \W$ such that $B_j$ is greedy for every $1 \leq j \leq i$, while $B_i, \ldots, B_k$ is a step walk.  In other words, $\W$ is greedy when noses (resp.\ hollows) are preferred until some point in which only steps follow.  It turns out that $\W_{N,i}$ and $\W_{H,i}$ are greedy paths.  The proof of this fact is analogous to those of Lemmas \ref{lem:greedy nose cycle}~and~\ref{lem:greedy hollow cycle}, yet we include it for the sake of completeness.

\begin{lemma}\label{lem:greedy walks}
  For any nose (resp.\ hollow) walk $\W_N$ of $\Syn$ there exists a greedy nose (resp.\ hollow) walk $\W_N'$ of $\Syn$ joining the same vertices such that either $\MAG(\W_N) < \MAG(\W_N')$ or $\MAG(\W_N) = \MAG(\W_N')$ and $\REM(\W_N) \leq \REM(\W_N')$.
\end{lemma}

\begin{proof}
  The proof is by induction on $|W_N| - p$ and $|W_N|$, where $p$ is the position of the first non-greedy vertex of $\W_N$.  The base case in which $p$ is greater than the position last nose (resp.\ hollow) of $\W_N$ is trivial. Suppose, then, that $B_1$ is the first non-greedy vertex of $\W_N$ and that $B_1$ appears before the last nose (resp.\ hollow) of $\W_N$.  Let $B_1 \to B$ the the nose (resp.\ hollow) from $B_1$ and consider the following cases.   

  \begin{description}
    \item[Case 1:] $\W_N$ is a nose walk.  Let $\W = B_1, \ldots, B_i$ be the shortest subpath of $\W_N$ such that either $B_{i-1} \to B_i$ is a nose or $B_i = B$, $\W'$ be the path formed by the nose $B_1 \to B$ followed by the step path from $B$ to $B_i$, and $\W_N'$ be the nose walk obtained from $\W_N$ by replacing $\W$ by $\W'$.  Clearly, the position of the first non-greedy vertex of $\W_N'$ is greater than $p$, thus, by induction, it suffices to show that $w(\W') \geq w(\W)$ for $w \in \{\MAG, \REM\}$ because $w(\W_N') = w(\W_N) - w(\W) + w(\W')$.  If $B_i = B$, then either $\W$ contains at least one $1$-step or $B_1 \to B_i$ is a $(-\Height)$-nose, thus $\MAG(\W') > \MAG(\W)$.  Otherwise, $B_1, B_{i-1}, B, B_{i}$ appear in this order in $\Syn$ when the steps are traversed from $B_1$.  As in Lemma~\ref{lem:greedy nose cycle}, it is not hard to see that $w(\W') = w(\W)$ when $B_1 \to B$ and $B_{i-1} \to B_i$ have equal jumps, while, as in Lemma~\ref{lem:greedy nose cycle}, only four cases remain otherwise.  All these cases are examined in the table below.

 \begin{center}
  \begin{tabular}{|c|c|c|c|c|c|c|c|}
    \hline
    $\Height(B_1)$ & $\Height(B_{i-1})$ & $\Height(B)$ & $\Height(B_{i})$ & $\MAG(\W)$ & $\MAG(\W')$ & $\REM(\W)$ & $\REM(\W')$ \\\hline
    $\Height-1$ & $\Height-1$ or $\Height$ & $\Height$   & $0$ & $-\Ratio$ & $-\Ratio$ & $-1$ & $-1$ \\
    $\Height-2$ & $\Height-1$              & $\Height-1$ & $0$ & $-1-\Ratio$ & $-1-\Ratio$ & $-1$ & $-1$  \\
    $\Height-1$ & $\Height$                & $0$         & $0$ & $-\Ratio$ & $-\Ratio$ & $-1$ & $-1$ \\
    $\Height$   & $0$                      & $0$         & $1$ & $-\Ratio$ & $-\Ratio$ & $-1$ & $-1$ \\\hline
  \end{tabular}
 \end{center}
 
    \item[Case 2:] $\W_N$ is a hollow walk.  Let $\W = B_1, \ldots, B_i$ be the shortest subpath of $\W_N$ such that $B_{i-1} \to B_i$ is a hollow.  If $B_i = B_j$ for some $1 \leq j < i$, then $\W_{ji} = B_j, \ldots, B_{i}$ is a cycle with exactly one $1$-step or $0$-hollow, thus $\MAG(\W_{ji}) < 0$.  So, the proof follows by induction on the hollow walk $\W_N' = \W_N \setminus \W_{ji}$ because $\MAG(\W_N') = \MAG(\W_j) + \MAG(\W_{ji})$ and the position of the first non-greedy hollow is at least $p$.  If $B_i \not\in \{B_1, \ldots, B_{i-1}\}$, then $B, B_{i}, B_1, B_{i-1}$ appear in this order in a traversal of the steps of $\Syn$ from $B$.  Let $\W'$ be the hollow path formed by the hollow $B_i \to B$ followed by the step path from $B$ to $B_i$ and observe that, as in Case~1, it suffices to prove that $w(\W') = w(\W)$ for $w \in \{\MAG, \REM\}$.  Moreover, both equalities hold when $B_1 \to B$ and $B_{i-1} \to B_{i}$ have equal jumps.  The equalities hold also when either $B_1 \to B$ or $B_{i-1} \to B_i$ is a $0$-hollow.  Indeed, if $B_1 \to B$ is a $0$-hollow, then $\Height(B) = \Height(B_{i}) = \Height(B_1)$ and $\Height(B_{i-1}) \neq \Height(B)$, while if $B_{i-1} \to B_i$ is a $0$-hollow, then $\Height(B_i) = \Height(B_1) = \Height(B_{i-1})$ and $\Height(B) \neq \Height(B_{i-1})$.  In both of these cases, $\MAG(\W) = \MAG(\W') = -1$ and $\REM(\W) = \REM(\W') = 0$.  Finally, when neither $B_1 \to B$ nor $B_{i-1} \to B_i$ are $0$-hollows and $B_1 \to B$ and $B_{i-1} \to B_i$ have different jumps, we are left with only six cases, as in the table below.

 \begin{center}
    \begin{tabular}{|c|c|c|c|c|c|c|c|}
      \hline
      $\Height(B_1)$ & $\Height(B_{i-1})$ & $\Height(B)$                 & $\Height(B_i)$ & $\MAG(\W)$   & $\MAG(\W')$  & $\REM(\W)$ & $\REM(\W')$ \\\hline
      $0$            & $1$                & $\Height$ or $\Height - 1$   & $0$            & $-1$        & $-1$        & $0$       & $0$        \\
      $0$            & $0$                & $\Height-1$                  & $\Height$      & $-1+\Ratio$ & $-1+\Ratio$ & $1$       & $1$        \\
      $\Height$      & $0$                & $\Height-1$                  & $\Height$      & $-1$        & $-1$        & $0$       & $0$        \\
      $\Height$      & $0$                & $\Height-1$                  & $\Height-1$    & $0$         & $0$         & $0$       & $0$        \\
      $\Height-1$    & $0$                & $\Height-2$                  & $\Height-1$    & $-1$        & $-1$        & $0$       & $0$        \\\hline
    \end{tabular}
  \end{center}
 \end{description} 
\end{proof}

Recall that our problem is to find the $\MAG$ and $\REM$ values for the path $\W_{N,i}$ satisfying \gnref{1}~and~\gnref{2}, for every $1 \leq i \leq n$.  We solve this problem in two phases.  Clearly, there exists a unique greedy nose path $\N$ beginning at $A_1$ that is maximal.  The first phase consist of traversing $\N = B_1, \ldots, B_k$ while $p(B_i) = (\MAG(B_1, \ldots, B_i), \REM(B_1, \ldots, B_i))$ is computed and stored for every $1 \leq i \leq k$.  In the second phase each step $A_{i-1} \to A_i$ of $\Syn$ is traversed while $q(A_i) = (\MAG(\W_{N,i}), \REM(\W_{N,i}))$ is computed and stored. By Lemma~\ref{lem:greedy walks}, $\W_{N,i}$ is a greedy nose path starting at $A_1$, thus $\W_{N,i}$ is equal to a subpath of $\N$ plus a (possibly empty) step path.  Consequently, there are only two possibilities for the last edge of $\W_{N,i}$ according to whether $A_i \in \N$ or not.  If $A_{i} \not\in \N$, then the last edge of $\W_{N,i}$ must be $A_{i-1} \to A_i$, thus $q(A_i) = q(A_{i-1}) + (\MAG(A_{i-1} \to A_i), \REM(A_{i-1} \to A_i))$.  Otherwise, the last edge could be $A_{i-1} \to A_i$ or the unique nose $A_j \to A_i$.  In the latter case $\W_{N,i}$ is a subpath of $\N$.  Thus, we can compute $\W_{N,i}$ by simply comparing the values $p(A_i)$ and $q(A_{i-1}) + (\MAG(A_{i-1} \to A_i), \REM(A_{i-1} \to A_i))$.  Note that both traversals cost $O(n)$ time.  The problem of finding $\W_{H,i}$ is analogous and it also costs $O(n)$ time.  

\begin{theorem}
  There is a unified certified algorithm that solves \Rep in $O(n)$ time.  
\end{theorem}

\subsection{Logspace construction of UCA models}
\label{sec:logspace algorithm}

In the full version of~\cite{KoblerKuhnertVerbitsky2012}, Kobler et al.\ ask whether it is possible to solve \Rep in (deterministic) logspace.  In this section we provide an affirmative answer to this question by showing that the algorithm of the previous section can be implemented so as to run in logspace.  Before doing so, we briefly discuss the logspace recognition of UCA graphs, for the sake of completeness.

Kobler et al.~\cite{KoblerKuhnertVerbitsky2012} show a logspace algorithm for the recognition of PCA graphs.  As a byproduct, their algorithm outputs a PCA model $\M$ with arcs $A_1 < \ldots < A_n$; we implement the algorithm by Kaplan and Nussbaum so as to run in logspace when $\M$ is given.  Let $\Syn_N$ be the subgraph of $\Syn$ obtained by removing all its hollows and all its steps that go from a vertex in which a nose starts.  All the vertices of $\Syn_N$ have out-degree $1$.  So, the nose ratio of each cycle $\W$ of $\Syn_N$ can be obtained in logspace by traversing $\W$ from each of its vertices.  By Lemma~\ref{lem:greedy nose cycle}, $\Ratio$ is the maximum among such ratios.  Then, taking into account that $\Syn_N$ can be easily computed in logspace from $\M$, we conclude that $\Ratio$ is obtainable in logspace.  An analogous algorithm can be used to compute the hollow ratio $\RATIO$ of $\M$ in logspace.  By Theorem~\ref{thm:tucker}, $\M$ is equivalent to a UCA model if and only if $\Ratio < \RATIO$.  The algorithm can output, also in logspace, the nose and hollow cycles with ratios $\Ratio$ and $\RATIO$, respectively.  These cycles provide a negative certificate that $\M$ is equivalent to no UCA models when $\Ratio \geq \RATIO$.

The logspace representation algorithm can be divided in two phases.  In the first phase, $\Red$ is build, while, in the second phase, the UCA model is obtained from a topological sort of its vertices.  Clearly, the problem of computing $\Red$ can be logspace reduced to querying which of the edges of $\Syn$ are redundant.  By \redref{1} and \redref{2}, the problem of testing if $A_i \to A_j$ is redundant is logspace reduced to that of finding $\IDist{\MAG}(A_1, A_i)$ and $(\IDist{\MAG} \circ \IDist{\REM})(A_1, A_i)$.  As stated in the previous section, this problem is reduced to that of computing $\MAG(\W_{N,i})$, $\REM(\W_{N,i})$, $\MAG(\W_{H,i})$, and $\REM(\W_{H,i})$ where $\W_{N,i}$ is the greedy nose path from $A_1$ to $A_i$ defined by \gnref{1}~and~\gnref{2}, while $\W_{H,i}$ is defined analogously.  By Lemma~\ref{lem:greedy nose cycle}, $\W_{N,i} = \N_{a,b}$ for some $a, b < n$, where $\N_{a,b}$ is the walk obtained by first traversing $a$ edges of the unique maximal greedy path $\N$ that begins at $A_1$, and then traversing $b$ steps.  So, by keeping the counters $a$ and $b$, we can compute the maximum among the values of $\MAG(\N_{a,b})$ and $\REM(\N_{a,b})$ for the paths $\N_{a,b}$ ending at $A_i$ with $a, b < n$.  Such a computation requires logspace, thus $\Red$ can be obtained in logspace.

Once $\Red$ is build, we could compute $\Unit(\Circ, \Len)$ by finding $\IDist{\SEP}(\Red, A_1, A_i)$ for every $1 \leq i \leq n$.  There is a major inconvenience with this approach: finding the longest path between two vertices of an acyclic digraph is a complete problem for the class of non-deterministic logspace problems.  To deal with this problem we could observe that $\Red$ is not only an acyclic digraph but also one with a rather particular structure.

\begin{theorem}\label{thm:toroidal}
  If $\M$ is a PCA model, then $\Syn$ is a toroidal digraph.
\end{theorem}

\begin{proof}
  As the theorem is implied by~\cite{Mitas1994}, we defer its proof to Section~\ref{sec:minimal uig} (see Corollary~\ref{cor:toroidal}).
\end{proof}

Toroidal acyclic digraphs are much simpler than general acyclic digraphs.  Yet,  up to this date, the best algorithms to compute their longest paths run in unambiguous logspace~\cite{LimayeMahajanNimbhorkarCJTCS2010}.  For this reason, the UCA model computed in the second phase is a variation of $\Unit(\Circ, \Len)$.  The key idea is to observe that the reachability problem for toroidal digraphs that have a unique vertex with in-degree $0$ can be solved in logspace~\cite{StoleeVinodchandran2012}.  That is, for $A_i, A_j \in V(\Red)$, the \emph{reachability algorithm} in~\cite{StoleeVinodchandran2012} outputs YES when there is a path from $A_i$ to $A_j$ in $\Red$.  Then, we can compute $\REACH(A_j) = |\{A_i \in V(\Red) \mid \text{there is a path from $A_i$ to $A_j$ in $\Red$}\}|$ in logspace for any given $A_j \in V(\Red)$.  The representation algorithm takes advantage of this fact by replacing $\CONST$ with the easier-to-find $\REACH$.  That is, the constructed UCA model $\Unit_{\REACH}$ is the $(\Circ, \Len)$-CA model with arcs $U_1, \ldots, U_n$ such that $\Circ = (\Len+1)(\Height + \Ratio) + \Extra$ for $\Extra = 4n$, $\Len+1 = \Ratio_2\Extra^2$, and 
\begin{align}
 s(U_j) = (\Len+1)(\Height(A_j) + \IDist{\MAG}(A_1, A_j)) + \Extra(\IDist{\REM} \circ \IDist{\MAG})(A_1, A_j) + 2\REACH(A_j) \label{eq:unit-reach}
\end{align}
for every $1 \leq j \leq n$.  Here, as in Theorem~\ref{thm:tucker} (\ref{thm:tucker:positive-certificate}), $\Ratio = \Ratio_1/\Ratio_2$, thus $\Circ$, $\Len$, and $s(U_j)$ are integers. By the previous discussion, $\Unit_\REACH$ is obtainable in logspace.  The fact that $\Unit_\REACH$ is equivalent to $\M$ follows from the next theorem.

\begin{theorem}
  Let $\M$ be a PCA model.  Then, $\Red$ is acyclic if and only if $\Unit_{\REACH}$ is an integer UCA model equivalent to $\M$.
\end{theorem}

\begin{proof}
  Suppose first that $\Red$ is acyclic and let $U_1, \ldots, U_n$ be the arcs of $\Unit_\REACH$ as in its definition.  Then, it suffices to show that:
  \begin{enumerate}[(a)]
    \item $s(U_j) \geq s(U_i) + \Len+1 - \Circ q$ for every nose $U_i \to U_j$ of $\Syn$,
    \item $s(U_j) \geq s(U_i) - \Len+1 + \Circ q$ for every hollow $U_i \to U_j$ of $\Syn$, and
    \item $s(U_j) \geq s(U_i) + 1 - \Circ q$ for every step $U_i \to U_j$.
  \end{enumerate}
  where $q \in \{0,1\}$ equals $0$ if and only if $A_i \to A_j$ is internal.  Take any edge $A_i \to A_j$ and, for the sake of notation, let:
  \begin{itemize}
    \item $\Delta x(i, j) = x(A_j) - x(A_i)$ for $x \in \{\Height, \REACH\}$, 
    \item $\Delta \MAG(i, j) = \IDist{\MAG}(A_1, A_j) - \IDist{\MAG}(A_1, A_i) - \MAG(A_i \to A_j)$, and
    \item $\Delta \REM(i, j) = (\IDist{\REM} \circ \IDist{\MAG})(A_1, A_j) - (\IDist{\REM} \circ \IDist{\MAG})(A_1, A_i) - \REM(A_i \to A_j)$.
  \end{itemize}
  By definition~\eqref{eq:unit-reach},
  \begin{align*}
    s(U_j) =& (\Len+1)(\Height(A_j) + \IDist{\MAG}(A_1, A_j)) + \Extra(\IDist{\REM} \circ \IDist{\MAG})(A_1, A_j) + 2\REACH(A_j) \\ 
           =& (\Len+1)(\Height(A_i) + \Delta\Height(i, j) + \IDist{\MAG}(A_1, A_i) + \MAG(A_i \to A_j) + \Delta{\MAG}(i, j)) +\\
            & \Extra((\IDist{\REM} \circ \IDist{\MAG})(A_1, A_i) + \REM(A_i \to A_j) + \Delta \REM(i, j)) +\\
            & 2\REACH(A_i) + 2\Delta\REACH(i, j) \\
           =& s(U_i) + (\Len+1)(\Delta\Height(i, j) + \MAG(A_i \to A_j)) + \Extra(\REM(A_i \to A_j)) + \varepsilon \tag{i} \label{eq:unit-reach-i}    
  \end{align*}
  where $\varepsilon = (\Len+1)\Delta\MAG(i, j) + \Extra \Delta\REM(i, j) + 2\Delta\REACH(i, j)$.
  
  Note that $\varepsilon \geq 2$.  Indeed, if $A_i \to A_j$ is redundant, then either $\Delta\MAG(i, j) > 0$ or $\Delta\MAG(i, j) = 0$ and $\Delta\REM(i, j) > 0$; thus $\varepsilon \geq 2$ as in Theorem~\ref{thm:tucker} (\ref{thm:tucker:negative-certificate}) $\Rightarrow$ (\ref{thm:tucker:positive-certificate}).  If $A_i \to A_j$ is not redundant, then $A_i \to A_j$ is an edge of $\Red$ and $\Delta\REACH(i, j) > 0$, while $\Delta\MAG(i, j) = \Delta\REM(i, j) = 0$.  Consequently, $\varepsilon \geq 2$ regardless of whether $A_i \to A_j$ is redundant or not.  Then, (a)--(c) follow by inspection, considering the 10 possible values for the jump of $A_i \to A_j$. For the sake of completeness, the table below sums up all these cases; recall that $\Circ = \Height(\Len+1) + \Ratio(\Len+1) + \Extra$.
  
  \begin{center}
    \begin{tabular}{|l|c|c|c|c|c|}
      \hline
      Type $A_i \to A_j$   & $q$ & $\Delta\Height(A_i, A_j)$ & $\MAG(A_i \to A_j)$ & $\REM(A_i \to A_j)$ & ${\rm \eqref{eq:unit-reach-i}} - \varepsilon$ \\\hline
      $1$-nose             & $0$ & $1$                       & $0$                 & $0$                 & $s(U_i) + \Len+1$ \\
      $(1-\Height)$-nose   & $1$ & $-\Height+1$              & $-\Ratio$           & $-1$                & $s(U_i) + \Len+1 - \Circ$ \\
      $(-\Height)$-nose    & $1$ & $-\Height$                & $1 - \Ratio$        & $-1$                & $s(U_i) + \Len+1 - \Circ$ \\
      $(-1)$-hollow        & $0$ & $-1$                      & $0$                 & $0$                 & $s(U_i) - \Len-1$ \\
      $0$-hollow           & $0$ & $0$                       & $-1$                & $0$                 & $s(U_i) - \Len-1$ \\
      $(\Height-1)$-hollow & $1$ & $\Height-1$               & $\Ratio$            & $1$                 & $s(U_i) - \Len-1 + \Circ$ \\
      $\Height$-hollow     & $1$ & $\Height$                 & $\Ratio-1$          & $1$                 & $s(U_i) - \Len-1 + \Circ$ \\    
      $0$-step             & $0$ & $0$                       & $0$                 & $0$                 & $0$ \\
      $1$-step             & $0$ & $1$                       & $-1$                & $0$                 & $0$ \\
      $(-\Height)$-step    & $1$ & $-\Height$                & $-\Ratio$           & $-1$                & $s(U_i) - \Circ$ \\\hline
    \end{tabular}
  \end{center}

  The converse follow from Theorem~\ref{thm:tucker}  (\ref{thm:tucker:equivalence}) $\Rightarrow$ (\ref{thm:tucker:negative-certificate}).
\end{proof}

\section{\texorpdfstring{$(a,b)$}{({\it a}, {\it b})}-independents and \texorpdfstring{$(x,y)$}{({\it x}, {\it y})}-circuits}
\label{sec:circuits and independents}

As stated, our algorithm outputs two cycles when $\M$ is not equivalent to a UCA model: a nose cycle $\W_N$ with ratio $\Ratio(\M)$ and a cycle $\W_H$ of $\Red$.  
As in the proofs of implications (\ref{thm:tucker:hollow-cycles}) $\Leftrightarrow$ (\ref{thm:tucker:cycles}) $\Leftrightarrow$ (\ref{thm:tucker:negative-certificate}),  $\W_H$ is a hollow cycle with a nonnegative length factor.  
Moreover, as in implication (\ref{thm:tucker:bound}) $\Leftrightarrow$ (\ref{thm:tucker:hollow-cycles}), $\RATIO(\W_H) \leq \Ratio(\W_N)$.
Note that, in principle, this certificate needs not be equal to the one in Section~\ref{sec:kaplan and nussbaum}, because $\W_H$ needs not be the hollow cycle with minimum ratio.  
Nevertheless, this certificate is somehow analogous to the one provided by the algorithm by Kaplan and Nussbaum.

Rigorously speaking, the certificate of Section~\ref{sec:kaplan and nussbaum} is neither equal to the one given by the algorithm by Kaplan and Nussbaum.  The former is a pair of nose and hollow cycles while the latter is a pair of $(a,b)$-independent plus $(x,y)$-circuit.  Nose cycles can contain more vertices than the corresponding $(a,b)$-independents while hollow cycles can contain more vertices than the corresponding $(x,y)$-circuits.  These added vertices are, nevertheless, redundant and can be eliminated from the certificate so as to obtain a minimal forbidden induced submodel as the negative certificate.  The purpose of this section is to describe the equivalence between nose (resp.\ hollow) cycles and $(a,b)$-independents (resp.\ $(x,y)$-circuits) and how to transform one into the other and vice versa.  We begin describing what are the $(a,b)$-independents and $(x,y)$-circuits.

For two arcs $A_i, A_j$ of a PCA model $\M$, we define the \emph{$ss$ arc of $A_i, A_j$} to be the arc $(s(A_i), s(A_j))$.  
For a sequence of arcs $\A = B_1, \ldots, B_k$, the \emph{$ss$ traversal of $\A$} is the family of arcs $\T$ that contains the $ss$ arc of $B_i, B_{i+1}$ for every $1 \leq i \leq k$ (where $B_1 = B_{k+1}$).  
The number of \emph{turns} of $\T$ is the number of its arcs that contain the point $0$ of $C(\M)$.  
In simple terms, the $ss$ traversal of $\A$ is obtained by traversing $C(\M)$ from $s(B_1)$ to $s(B_2)$ to \ldots to $s(B_k)$ to $s(B_1)$, while its number of turns is the number of complete loops to the circle in such a traversal.

An $(a,b)$-independent of a PCA model $\M$ is a sequence of arcs $\A = B_1, \ldots, B_a$ such that $s(B_{i+1}) \not \in B_i$ for every $1 \leq i \leq a$ and whose $ss$ traversal takes $b$ turns.  Similarly, an $(x,y)$-circuit is a sequence of arcs $B_1, \ldots, B_x$ such that $s(B_{i+1}) \in B_i$ for every $1 \leq i \leq x$ and whose $ss$ traversal takes $y$ turns.  Note that $x > 2y$ as no pair of arcs of $\M$ cover the circle.  An $(a,b)$-independent is \emph{maximal} when $a/b$ is maximum and $a,b$ are relative primes, while an $(x,y)$-circuit is \emph{minimal} when $x/y$ is minimum and $x,y$ are relative primes.  As we shall shortly see, statement (\ref{thm:tucker:equivalence}) $\Leftrightarrow$ (\ref{thm:tucker:bound}) of Theorem~\ref{thm:tucker} is equivalent to the following theorem by Tucker.

\begin{theorem}[\cite{TuckerDM1974}]
  A PCA model $\M$ is equivalent to an UCA model if and only if $a/b < x/y$ for every maximal $(a,b)$-independent and every minimal $(x,y)$-circuit.
\end{theorem}

Say that an $(a,b)$-independent $\A = B_1, \ldots, B_a$ is \emph{standard} when $s(B_i)$ is immediately preceded by an ending point in $\M$, for every $1 \leq i \leq a$.  Note that if $s(B_i)$ is preceded by the beginning point of an arc $A$, then $B_1, \ldots, B_{i-1}, A, B_{i+1}, \ldots, B_{a}$ is also an $(a, b)$-independent of $\M$.  Consequently, $\M$ has an $(a,b)$-independent if and only if it has a standard $(a, b)$-independent.

There is a one-to-one correspondence between the standard $(a,b)$-independents of $\M$ and the nose circuits of $\Syn$, as follows.  Let $\A = B_1, \ldots, B_a$ be a standard $(a, b)$-independent and $\W_i$ be the step path of $\Syn$ that goes from $B_i$ to $B_i'$, where $B_i'$ is the arc whose ending point immediately precedes $s(B_{i+1})$.  Clearly, $B_i' \to B_{i+1}$ is a nose of $\Syn$, thus $\W(\A) = \W_1, \W_2, \ldots, \W_a$ is a nose circuit of $\Syn$.  Conversely, if $\W$ is a nose circuit, and $B_1' \to B_1, \ldots, B_a' \to B_a$ are its noses, then $\A(\W) = B_1, \ldots, B_a$ is an standard $(a, b)$-independent for some $b$.  It is not hard to see that $\A(\W(\A)) = \A$ and $\W(\A(\W)) = \W$, thus the correspondence is one-to-one.  

Observe that the number of turns $b$ in the $ss$ traversal of $\A$ is precisely the number of external noses and steps of $\W = \W(\A)$.  In other words, 
\[b = \Noses_{-\Height} + \Noses_{1-\Height} + \Steps_{-\Height}.\]  
Similarly, the number $a$ of arcs of $\A$ equals the number of noses of $\W$; by \eqref{eq:walk jump}, 
\[a = \Height(\Noses_{-\Height} + \Noses_{1-\Height} + \Steps_{-\Height}) + \Noses_{-\Height} - \Steps_1 =  \Height b + \Noses_{-\Height} - \Steps_1.\]
Hence, $a/b = \Height + \Ratio(\W)$.

A similar analysis holds for $(x, y)$-circuits.  Say that an $(x, y)$-circuit $\A = B_1, \ldots, A_x$ is \emph{standard} when $s(B_i)$ immediately precedes an ending point $t(B_i')$ in $\M$.  Note that $\M$ contains an $(x,y)$-circuit if and only if it contains an standard $(x,y)$-circuit; in such circuit, $B_i \to B_i'$ is a hollow of $\Syn$.  Then, $\A$ is in a one-to-one correspondence with the hollow circuit $\W = \W(\A)$ that goes through $\W_1, \W_2, \ldots, \W_x$ where each $\W_i$ is the step path going from $B_{i}'$ to $B_{i+1}$.  As before, the number of turns in the $ss$ traversal of $\A$ is the number of external hollows minus the number of external steps, i.e.,
\[y = \Hollows_\Height + \Hollows_{\Height-1} - \Steps_{-\Height},\] 
while the number $x$ of arcs in $\A$ is its number of hollows; by \eqref{eq:walk jump}
\[x = \Height(\Hollows_{\Height} + \Hollows_{\Height-1} - \Steps_{-\Height}) + \Hollows_{\Height} + \Hollows_0 + \Steps_1.\] Hence, $x/y = \Height + \RATIO(\W)$.

Clearly, we can obtain $\W(\A)$ in $O(n)$ time for any standard $(a,b)$-independent (resp.\ $(x,y)$-circuit) $\A$, and vice versa.  Moreover, note that, as stated, $a/b \geq x/y$ if and only if $\Ratio \geq \RATIO$.  We summarize this section in the next theorem.  

\begin{theorem}
  A PCA model $\M$ contains an $(a,b)$-independent (resp.\ $(x,y)$-circuit) $\A$ if and only if $\Syn$ contains a nose (resp.\ hollow) circuit $\W$ with ratio $a/b - \Height$ (resp.\ $x/y - \Height$).  Moreover, such a circuit $\W$ of $\Syn$ can be obtained in $O(n)$ time when $\A$ is given as input.  Conversely, $\A$ can be obtained in $O(n)$ time when $\W$ is given as input.
\end{theorem}

\section{Minimal UCA and UIG models}
\label{sec:minimal uca}

Theorem~\ref{thm:separation constraints} gives us a procedure to check if $\M$ is equivalent to a $\Descriptor$-CA model, when $\Descriptor$ is given as input.  However, not much is known about the sets of feasible values $\Circ$ and $\Len$.  In this aspect, unit circular-arc models are much less studied than unit interval models.    In this section we prove that every UCA model admits an equivalent \emph{minimal} UCA model.  Minimal UCA models are a generalization of minimal UIG models, as defined by Pirlot~\cite{PirlotTaD1990}.  An $(\Len, \Thres, \BegDist)$-IG model with arcs $A_1 < \ldots < A_n$ is \emph{$(\Thres, \BegDist)$-minimal} when
\begin{enumerate}[({min-uig}$_1$)]
  \item $\Len \leq \Len'$, and \label{def:minuig-1}
  \item $s(A_i) < s(A_i')$ for every $1 \leq i \leq n$, \label{def:minuig-2}
\end{enumerate}
\newcommand{\minuigref}[1]{\itemref{\text{min-uig}}{minuig}{#1}}
for every equivalent $(\Len', \Thres, \BegDist)$-IG model.

Condition~\minuigref{2} as expressed above does not make much sense for UCA models, as there is not a natural left-to-right order of the arcs; the $0$ point of the circle is just a denotational tool.  However we can translate condition $s(A_n) < s(A_n')$ by asking the circumference of the circle to be minimized.  With this in mind, we say that a $(\Circ, \Len, \Thres, \BegDist)$-CA model $\M$ is \emph{$(\Thres, \BegDist)$-minimal} when
\begin{enumerate}[({min-uca}$_1$)]
  \item $\Len \leq \Len'$, and  \label{def:minuca-1}
  \item $\Circ \leq \Circ'$,\label{def:minuca-2}
\end{enumerate}
\newcommand{\minucaref}[1]{\itemref{\text{min-uca}}{minuca}{#1}}
for every equivalent $(\Circ', \Len', \Thres, \BegDist)$-CA model.  The fact that every UCA model is equivalent to a minimal UCA model follows from the next lemma.

\begin{lemma}\label{lem:minimal model}
 If $\M$ is equivalent to a $(\Circ, \Len + y, \Thres, \BegDist)$-CA model and a $(\Circ + x, \Len, \Thres, \BegDist)$-CA model for $x,y \geq 0$, then $\M$ is also equivalent to a $(\Circ + a, \Len +b, \Thres, \BegDist)$-CA model, for every $0 \leq a \leq x$ and $0 \leq b \leq y$.
\end{lemma}

\begin{proof}
 For the sake of notation, write $\langle a', b'\rangle$ to denote the UCA descriptor $(\Circ + a, \Len + b, \Thres, \BegDist)$, for every $0 \leq a' \leq x$ and $0 \leq b' \leq y$.  Suppose, to obtain a contradiction, that $\M$ is equivalent to no $\langle a, b\rangle$-CA model for some $1 \leq a \leq x$ and $1 \leq b \leq y$.  Then, by Theorem~\ref{thm:separation constraints}, $\SEP_{\langle a, b\rangle}(\W) > 0$, $\SEP_{\langle x, 0 \rangle}(\W) \leq 0$, and $\SEP_{\langle 0, y\rangle}(\W) \leq 0$ for some cycle $\W$ of $\Syn$.  By~\eqref{eq:ratios}, there exists $\Extra \geq 0$ such that
 \begin{align*}
   \Circ + a' &= (\Len + y + \Thres)(\Height + \Ratio) + \Extra + a' = (\Len + b' + \Thres)(\Height + \Ratio) + (y-b')(\Height + \Ratio) + \Extra + a'  \tag{i}
 \end{align*}
 for every $0 \leq a' \leq x$ and $0 \leq b' \leq y$.  Thus, by~\eqref{eq:sep},
 \begin{align*}
  \SEP_{\langle a, b \rangle} &= (\Len+b+\Thres)\MAG + ((y-b)(\Height + \Ratio) + \Extra + a)\REM + \CONST  \tag{ii} > 0\\
  \SEP_{\langle 0, y \rangle} &= (\Len+y+\Thres)\MAG + \Extra\REM + \CONST  \leq 0 \tag{iii}\\
  \SEP_{\langle x, 0 \rangle} &= (\Len+\Thres)\MAG + (y(\Height+\Ratio)+\Extra+x)\REM + \CONST \leq 0 \tag{iv}.
 \end{align*}
 Recall that $\MAG < 0$ by Theorem~\ref{thm:tucker}.  Then, as $0 < {\rm (ii)} - {\rm (iv)}$, we obtain that (v) $\REM \leq -1$ (recall $\REM \in \mathbb{Z}$), while (vi) $(\Height+\Ratio)\REM > \MAG$ follows by $0 < {\rm (ii)} - {\rm (iii)}$ and (v).  Then,
 \begin{align}
   0 <& (\Len+b+\Thres)\MAG + ((y-b)(\Height + \Ratio) + \Extra + a)\REM + \CONST \tag{by (ii)} \\
     <&  (\Len+b+\Thres)(\Height+\Ratio)\REM + ((y-b)(\Height + \Ratio) + \Extra + a)\REM + \CONST \tag{by (vi)} \\
     =& (\Circ+a)\REM + \CONST \tag{by (i)} \\
     \leq&  -\Circ - a + \CONST \tag{by (v)}
 \end{align}
 This is impossible, because $\Circ \geq \max\{2\Thres, \Thres+\BegDist\}n$ as all the extremes of the $\langle 0, y\rangle$-CA model equivalent to $\M$ are separated by $\Thres$ and each of its $n$ beginning points is separated from the next by $\Thres + \BegDist$, while $a \geq 0$ and $\CONST \leq \max\{2\Thres, \Thres+\BegDist\}n$ by definition~\eqref{eq:const}. 
\end{proof}

\begin{theorem}
  Every UCA graph admits a $(\Thres, \BegDist)$-minimal UCA model for every $\Thres, \BegDist \in \mathbb{Q}$.
\end{theorem}

For the rest of this section, we restrict ourselves to the case in which $\Thres$ and $\BegDist$ are integers.  By~\eqref{eq:nose ratio}, $\Ratio = 0$ for every UIG model $\M$, thus, by~\eqref{eq:sep}, $\SEP(\W) \in \mathbb{N}$ if and only if $\Len$ and $\Circ$ are integers.  We obtain, therefore, the following corollary that was first proved by Pirlot~\cite{PirlotTaD1990}.

\begin{corollary}
  Every $(\Thres, \BegDist)$-minimal UIG model is integer for every $\Thres, \BegDist \in \mathbb{N}$
\end{corollary}

The above corollary holds for UCA models with an integer nose ratio as well.  However, we were not able to prove or disprove the above corollary for the general case.  For this reason, we say that an integer $(\Circ, \Len, \Thres, \BegDist)$-CA model $\M$ is \emph{$(\mathbb{N}, \Thres, \BegDist)$-minimal} when if satisfies \minucaref{1} and \minucaref{2} for every integer $(\Circ', \Len', \Thres, \BegDist)$-CA model.  

A natural algorithmic problem is \MinUCA in which we ought to find a $(\mathbb{N}, \Thres, \BegDist)$-minimal UCA model equivalent to an input $(\Circ, \Len, \Thres, \BegDist)$-CA model $\M$.  A simple solution is to apply Theorem~\ref{thm:separation constraints} for every $1 \leq \Len^* \leq \Len$ and every $1 \leq \Circ^* \leq \Circ$ with a total cost of $O(\Len^*\Circ^* n^2)$ time.  We can easily improve this algorithm by replacing the linear search of $\Circ^*$ with a binary search.

\begin{corollary}\label{cor:circle range}
  Let $\M$ be a PCA model.  If $\SEP_{(\Circ,\Len,\Thres,\BegDist)}(\W) > 0$ for some cycle $\W$ of $\Syn$, then $\M$ is not equivalent to a $(\Circ + \sg . x, \Len - y, \Thres, \BegDist)$-CA model, for every $x,y \geq 0$, where $\sg$ is the sign of $\REM(\W)$.
\end{corollary}

\begin{proof}
  Let $\Extra$ be such that $\Circ = (\Len+\Thres)(\Height + \Ratio) + \Extra$; note that $\Extra$ needs not be positive.  By~\eqref{eq:sep}, 
  \[\SEP_{(\Circ+\sg x, \Len,\Thres,\BegDist)} = (\Len+\Thres)\MAG + (\Extra + \sg . x)\REM + \CONST = \sg . x . \REM + \SEP_{(\Circ, \Len,\Thres,\BegDist)} > 0.\]
  Then, by Theorem~\ref{thm:separation constraints}, $\M$ cannot be equivalent to a $(\Circ + \sg . x, \Len, \Thres,\BegDist)$-CA model.
\end{proof}

Corollary~\ref{cor:circle range} provides us with a somehow efficient algorithm to binary search the minimum $\Circ^* \in\mathbb{N}$ such that $\M$ is equivalent to a $(\Circ^*, \Len, \Thres, \BegDist)$-CA model, when $\Len, \Thres, \BegDist \in \mathbb{N}$ are given as input.  By definition $\Thres < \Len$, while we assume $\BegDist + \Thres < \Len$ as otherwise $G(\M)$ has no edges and the problem is trivial. The idea of the algorithm is simply to assume that such a value $\Circ^*$ exists and belongs to some range $[a, \Circ]$; initially $a = 0$ and $\Circ = n(\Len+1)$.  Then, we query if $\M$ is equivalent to some $(b,\Len, \Thres,\BegDist)$-CA model, where $b \in \mathbb{N}$ is the middle of $[a,\Circ]$.  If affirmative, then $\Circ^* \in [a,b]$ by definition.  Otherwise, we search some cycle $\W$ with $\SEP_{(b,\Len,\Thres,\BegDist)} > 0$.  By Corollary~\ref{cor:circle range}, $\Circ^* \in [0,b)$ if $\REM(\W) \geq 0$, while $\Circ^* \in (b,\Circ]$ otherwise.  Regardless of whether $\Circ^*$ exists or not, this algorithm requires $O(\log (n\Len))$ queries.

Every time we need to query if $\M$ is equivalent to a $(b,\Len,\Thres,\BegDist)$-CA model, we solve \uRep as in Section~\ref{sec:synthetic graph:bounded algorithm}.  Since $\Len^* = O((\Thres + \BegDist)n^2)$~\cite{LinSzwarcfiterSJDM2008}, we conclude that the total time required to obtain $(\Circ^*, \Len^*)$, thus solving \MinUCA, is $O((\Thres + \BegDist)n^4\log (n+\Thres+\BegDist))$.

\subsection{Minimal UIG models}
\label{sec:minimal uig}

Pirlot proved in~\cite{PirlotTaD1990} that every UIG model $\M$ is equivalent to a $(1,0)$-minimal UIG model.  However, it was Mitas who showed that such a model can be found in linear time by transforming $\M$ into an equivalent UIG model $\M^*$~\cite{Mitas1994}.  Unfortunately, her proof has a flaw that invalidates the minimality arguments.  Though $\M^*$ is equivalent to $\M$, it needs not be $(1,0)$-minimal.  On the other hand, the algorithm in the previous section can be implemented so as to run in $O(n^2\log n)$ time when applied to $\M$.  (Below we discuss condition why \minuigref{2} is satisfied by the unit interval model so obtained.)  In this section we briefly describe Mitas' algorithm and its counterexample, and propose a patch.  The obtained algorithm runs in $O(n^2)$ time,  and its a compromise between Mitas' algorithm and the algorithm in the previous section.

Let $\M$ be a PIG model with arcs $A_1 < \ldots < A_n$.  By definition, no arc of $\M$ crosses $0$, thus $\Syn$ has no external hollows.  Similarly, external noses and steps are redundant in $\Syn$ for testing if $\M$ is equivalent to a $\Len$-IG model, as $\Circ = \infty$.  Therefore, we assume that $\Syn$ has only five types of edges, namely $1$-noses, $0$- and $1$-steps, and $0$- and $(-1)$-hollows.  Mitas identifies two special vertices of $\Syn$ for each height value.  A \emph{leftmost} vertex is a vertex $A_i$ such that either $i = 1$ or $\Height(A_i) = \Height(A_{i-1}) + 1$, while a \emph{rightmost} vertex is a vertex $A_j$ such that either $j = n$ or $\Height(A_{j+1}) = \Height(A_j) + 1$.  

\begin{figure}[th]
  \begin{minipage}{\linewidth}\centering
    \includegraphics{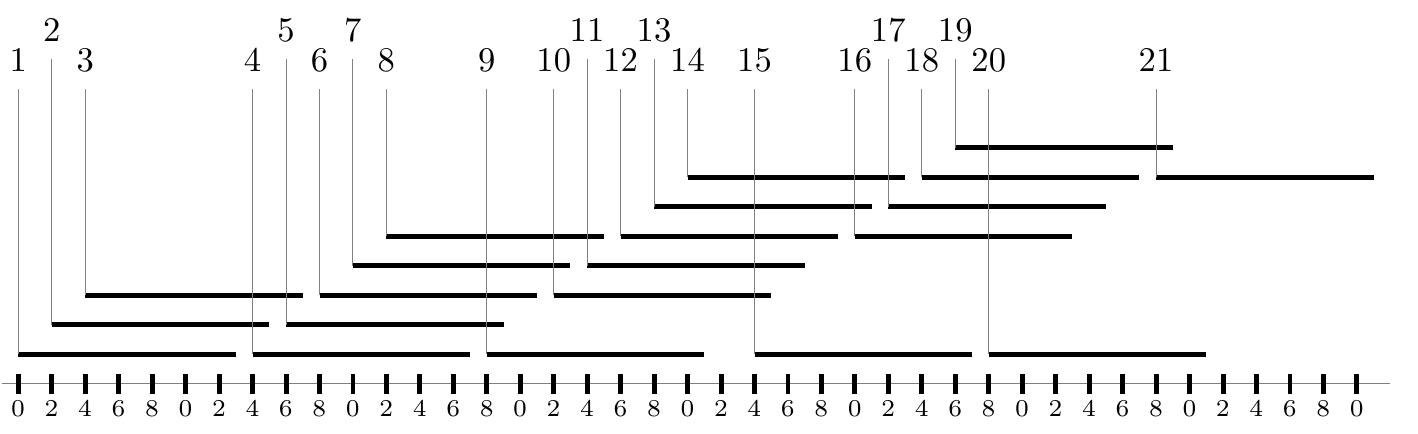}
    \subcaption{A $13$-UIG model $\M$.}
  \end{minipage}
  
  \begin{minipage}{\linewidth}\centering
    \includegraphics{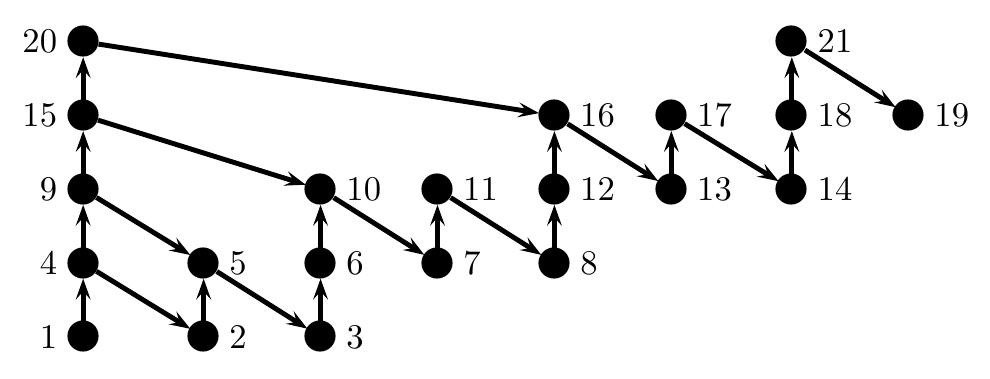}
    \subcaption{The canonical drawing of $\RedMitas(\M)$; $0$-steps are not shown for the sake of exposition.}
  \end{minipage}

  \caption{Counterexample to \eqref{eq:mitas bug}: $\IDist{\SEP}(\RedMitas, A_1, A_{19}) - \IDist{\SEP}(\RedMitas, A_1, A_{15}) = 14$ while in $\Syn$ the maximum cycle has length $13$.}\label{fig:mitas counterexample}
\end{figure}

Suppose $\M$ is equivalent to some $\Len$-IG model and let $\Unit = \Unit(\infty, \Len)$ be as defined in Theorem~\ref{thm:separation constraints}, but replacing $\BoundedSyn$ and $A_0$ with $\Syn$ and $A_1$, respectively.  That is, the arc $U_i$ corresponding to $A_i$ begins at $s(U_i) = \IDist{\SEP}(A_1, A_i)$ for every $1 \leq i \leq n$.  
For $x \in \{1, \ldots, \Height\}$, let $A_i$ and $A_j$ be leftmost and rightmost vertices with $\Height(A_i) = \Height(A_j) = x$, and $\W_x$ be a path from $A_i$ to $A_j$.  By definition of $\U$, it follows that $s(U_j) - s(U_i) \geq \SEP(\W_x)$.  That is,
\begin{equation}
\IDist{\SEP}(A_1, A_j) - \IDist{\SEP}(A_1, A_i) \geq \SEP(\W_x). \label{eq:mitas correct}
\end{equation}
Mitas' key idea is to take $\Len$ so that \eqref{eq:mitas correct} holds by equality when $\SEP(\W_x)$ is maximum (recall $\SEP$ depends on $\Len$).  The flaw, however, is that she discards $0$-hollows and $1$-steps before solving~\eqref{eq:mitas correct}.  

To make the above statement more precise, let $\RedMitas(\M)$ be the digraph obtained from $\Syn(\M)$ by removing all the $0$-hollows and $1$-steps, and $\W_x^\RedMitas$ be a path from $A_i$ to $A_j$ in $\RedMitas$ (as usual we drop the parameter $\M$ from $\RedMitas$).  Mitas claims (Theorem~5 in~\cite{Mitas1994}, although using a different terminology) that 
\begin{equation}
 \IDist{\SEP}(\RedMitas, A_1, A_j) - \IDist{\SEP}(\RedMitas, A_1, A_i) = \SEP(\W_x^\RedMitas). \label{eq:mitas bug}
\end{equation}
when $\SEP(\W_x^\RedMitas)$ is maximum and $\Len$ is minimum.  Figure~\ref{fig:mitas counterexample} shows a counterexample to this fact.  The inconvenient is that $\IDist{\SEP}(\Syn, A_1, A_i)$ is greater than $\IDist{\SEP}(\RedMitas, A_1, A_i)$ when every maximum path from $A_1$ to $A_i$ contains the $0$-hollow or $1$-step ending at $A_i$. 

Equation~\eqref{eq:mitas bug} is fundamental for keeping the time and space complexity low.  The main observation is that $\RedMitas$ is acyclic whereas $\Syn$ is not.

\begin{lemma}[\cite{Mitas1994}]\label{lem:RedMitas acyclic}
  Digraph $\RedMitas(\M)$ is acyclic, for any UIG model $\M$.
\end{lemma}

Given that $\RedMitas$ is acyclic, we can compute the \emph{column} that every arc of $\M$ occupies in a pictorial description of $\RedMitas$.  The \emph{column} of $A_1$ is $\Column(A_1) = 0$, while, for every $1 < i \leq n$, the \emph{column} of $A_i$ is:
\begin{align}
  \Column(A_i) = \max\left\{
    \begin{array}{r}
       \Column(N) + \varepsilon \\ \Column(H) + 1 \\ \Column(S) + 1
    \end{array} \middle| 
    \begin{array}{l}
      N \to A_i \text{ is a $1$-nose } \\
      H \to A_i \text{ is a $0$-step } \\
      S \to A_i \text{ is a $-1$-hollow } \\
    \end{array}\right\}\label{eq:column}
\end{align}
for a small enough $\varepsilon$ (say $\varepsilon \ll 1/n$); obviously, if $A_i$ is not the end of a nose (resp.\ hollow, step), then the corresponding value in the above equation is $0$.  It is easy to see, by the existence of $0$-steps, that $\Column(I_i) \leq \Column(I_k) \leq \Column(I_j)$ when $A_i$ and $A_j$ are the leftmost and rightmost with height $x$, for every $A_i < A_k < A_j$.  In Figure~\ref{fig:mitas counterexample}, each vertex $A$ of $\Syn$ occupies the coordinate $(\Column(A), \Height(A))$ on the plane, for some imperceptible $\varepsilon$, while each directed edge is a straight arrow.  This pictorial description, which we call the \emph{canonical drawing} of $\RedMitas$, was proposed by Mitas and it is quite useful for simplifying some geometrical arguments.  The reason is that this drawing is a plane digraph; we include a proof of this fact as it is not completely explicit in~\cite{Mitas1994}.

\begin{theorem}[see~\cite{Mitas1994,PirlotVincke1997}]\label{thm:canonical drawing}
  The canonical drawing of $\RedMitas$ is a plane digraph.
\end{theorem}

\begin{proof}
  Suppose, to obtain a contradiction, that the canonical drawing of $\RedMitas$ it is not a plane graph.  Then, there are two crossing straight lines that correspond to the edges $A_i \to A_j$ and $A_x \to A_y$ with $\Height(A_i) \leq \Height(A_x)$.  By definition, $\RedMitas$ has only $1$-noses, $(-1)$-hollows and $0$-steps, while every vertex $A$ is positioned in $(\Column(A), \Height(A))$.  Hence, it follows that $A_i \to A_j$ is a $1$-nose, $A_x \to A_y$ is a $(-1)$-hollow, and $A_i < A_y < A_x < A_j$.  But this configuration is impossible because it implies that $t(A_i)s(A_j)$ and $s(A_x)t(A_y)$ are consecutive, while $t(A_i) < t(A_y)$ and $s(A_x) < s(A_j)$.
\end{proof}

\begin{corollary}[Theorem~\ref{thm:toroidal}]\label{cor:toroidal}
  If $\M$ is a PCA graph, then $\Syn$ is a toroidal digraph.
\end{corollary}

\begin{proof}
  A torus can be obtained from a rectangle by first pasting its north and south borders together, and then pasting the east end of the obtained cylinder with its west end.  Thus, it suffices to show how to draw $\Syn$ into a rectangle allowing some edges to escape from the north (resp.\ east) into the south (resp.\ west).  Let $\RedMitas$ be obtained from $\Syn$ by removing all the external edges, plus $1$-steps and $0$-hollows.  To draw $\Syn$, first copy the canonical drawing of $\RedMitas$ into the rectangle.  Then, draw all the $0$-hollows and $1$-steps so that they escape through the east, all the $\Height$-hollows and $(-\Height)$-noses so as to run through the north, and the $(-\Height)$-step and all the $(\Height-1)$-hollows and $(1-\Height)$-noses by going through the north first and then through the east.  It is not hard to see that such a drawing is always possible. 
\end{proof}

In the next lemma we take advantage of the canonical drawing to prove that every cycle of $\Syn$ contains exactly one $0$-hollow or $1$-step.  Pirlot also studies the shape of the cycles of $\Syn$~\cite{PirlotTaD1990}, but without taking advantage of Mitas' canonical drawing.  For the next lemma, recall that $\MAG(\W) = -\Hollows_0(\W) - \Steps_1(\W)$ for any cycle $\W$.

\begin{lemma}[{\cite[Proposition~2.11]{PirlotTaD1990}}]\label{lem:one backedge}
  If $\M$ is a PIG model, then $\MAG(\W) = -1$ for any cycle $\W$ of $\Syn$. 
\end{lemma}

\begin{proof}
  Note that $\MAG = -\Hollows_0 - \Steps_1$ \eqref{eq:mag}, hence, by Lemma~\ref{lem:RedMitas acyclic}, $\MAG < 0$.  Suppose, to obtain a contradiction, that $\MAG < -1$.  Then, $\W$ has a subpath $B_1, \ldots, B_j$ with no $0$-hollows nor $1$-steps such that $B_0 \to B_1$ and $B_j \to B_{j+1}$ each is either a $0$-hollow or a $1$-step of $\W$.  Among all such possible paths, take $\W$ so that $\Height(B_1)$ is maximum.  Note that $B_0 \neq B_j$, thus $\W$ has another path $B_{-k}, \ldots, B_{0}$ such that $B_{-k}$ is its unique leftmost vertex.  By the maximality of $\Height(B_1)$, it follows that $\Height(B_0) \geq \Height(B_1)-1 \geq \Height(B_{-k})$ while, since $\Height(B_j) \leq \Height(B_{j+1}) \leq \Height(B_1) - 1 \leq \Height(B_0)$ and $B_0 \neq B_j$, it follows that $\Height(B_0) > \Height(B_j)$.
  
  Call ${\rm Gr}^+$ to the curve that results by traversing $B_1, \ldots, B_j$ in the canonical drawing of $\RedMitas$.  Note that ${\rm Gr}^+$ is indeed the graph of a continuous function on $\mathbb{R} \to \mathbb{R}$ because $\Column(B_{i+1}) >  \Column(B_{i})$ for every $1 \leq i < j$ by \eqref{eq:column}.  Similarly, the curve ${\rm Gr}^-$ that results by traversing $B_{-k}, \ldots, B_0$ in the canonical drawing of $\RedMitas$ is also the graph of a continuous function.  Since $\Height(B_{i+1}) = \Height(B_i) \pm 1$ for every $i \in \{-k, \ldots, j\} \setminus 0$, it follows that ${\rm Gr}^+$ contains a vertex with height $x$ for every $\Height(B_j) \leq x \leq \Height(B_1)$ and ${\rm Gr}^-$ contains a vertex with height $x$ for every $\Height(B_{-k}) \leq x \leq \Height(B_0)$.
  Then, taking into account that $B_1$ and $B_{-k}$ are leftmost vertices with $\Height(B_1) > \Height(B_{-k})$ and $B_j$ and $B_0$ are rightmost vertices with $\Height(B_0) > \Height(B_j)$, we obtain that ${\rm Gr}^+$ and ${\rm Gr}^-$ intersect.  Hence, by Theorem~\ref{thm:canonical drawing}, $B_1, \ldots, B_j$ and $B_{-k}, \ldots, B_{0}$ have a nonempty intersection, which implies that $\W$ is not a cycle.
\end{proof}

By Lemma~\ref{lem:one backedge} and \eqref{eq:sep}, $\SEP_\Descriptor(\W) = \CONST(\W,\Thres,\BegDist) - \Len - 1$ for every cycle $\W$.  Then, by Theorem~\ref{thm:separation constraints} and Lemma~\ref{lem:one backedge}, the minimum $\Len^*$ such that $\M$ is equivalent to an $(\Len^*, \Thres, \BegDist)$-IG model is 
\begin{align*}
 \Len^*+1 =& \max\{\CONST(\W,\Thres, \BegDist)  \mid \W \text{ is a cycle of } \Syn\} \\
          =& \max\{\CONST(\W,\Thres, \BegDist)  \mid \W \text{ is a path } B_1, \ldots, B_j \text{ of $\RedMitas$ for a $1$-step or $0$-hollow $B_j \to B_1$}\} \\
          =& \max\{\IDist{\CONST_{\Thres, \BegDist}}(\RedMitas, A_i, A_j)  \mid A_j \to A_i \text{ is either a $1$-step or $0$-hollow}\}.
\end{align*}
Since $\RedMitas$ is acyclic, we can compute $\IDist{\CONST_{\Thres, \BegDist}}(\RedMitas, A_i, A_j)$ in $O(n)$ time and space for any given $1$-step or $0$-hollow $A_j \to A_i$ of $\Syn$.  Then, $\Len^*$ is obtained in $O(\Height n)$ time.

Once $\Len^*$ has been obtained, $\Unit = \Unit(\infty, \Len^*, \Thres, \BegDist)$ can be constructed in $O(n^2)$ time and linear space as in Section~\ref{sec:synthetic graph:bounded algorithm}.  We claim that $\Unit$ is a $(\Thres, \BegDist)$-minimal UIG model.  Indeed, $\Unit$ satisfies \minuigref{1} by the minimality of $\Len^*$.  To see that $\Unit$ satisfies \minuigref{2}, consider any path $\W$ of $\Syn$ from $A_1$ to $A_j$.  Note that $\Height(A_j) \geq \Steps_{1} + \Hollows_{0} = -\MAG$ because no leftmost vertex is traversed twice by $\W$.  Therefore, by \eqref{eq:sep}, 
\[\SEP_{(\infty, \Len, \Thres, \BegDist)} = (\Len+1)(\Height(A_j)+\MAG) + \CONST \geq (\Len^*+1)(\Height(A_j)+\MAG) + \CONST = \SEP_{(\infty, \Len^*, \Thres, \BegDist)}\]
for any $\Len \geq \Len^*$.  Consequently, since $s(A_j) \geq \IDist{\SEP_{(\infty, \Len, \Thres, \BegDist)}}(A_1, A_j)$ in any $(\Len, \Thres, \BegDist)$-UIG model equivalent to $\Unit$, it follows that $\Unit$ satisfies \minuigref{2} as well.  We conclude that $O(n^2)$ time and linear space suffices to solve the \emph{minimal UIG representation} (\MinUIG) problem in which $\M$ and $\Thres, \BegDist \in \mathbb{Q}_{\geq 0}$ are given and a $(\Thres, \BegDist)$-minimal UIG model equivalent to $\M$ must be generated.

\begin{theorem}\label{thm:minimal algorithm}
  \MinUIG can be solved in $O(n^2)$ time and linear space, for any $\Thres, \BegDist \in \mathbb{Q}$.
\end{theorem}

\section{Powers of paths and cycles}
\label{sec:powers}

Powers of paths and cycles are intimately related to UIG and UCA graphs, respectively.  For any graph $G$, its \emph{$k$-th power} $G^k$ is the graph obtained from $G$ by adding an edge between $v$ and $w$ whenever there is a path in of length at most $k$ joining them.  In this section we write $P_q$ and $C_q$ to denote the path and cycle graphs with $q$ vertices, respectively.  Lin et al.~\cite{LinRautenbachSoulignacSzwarcfiterDAM2011} noted that $G$ is a UCA (resp.\ UIG) graph if and only if $G$ is an induced subgraph of $C_q^k$ (resp.\ $P_q^k$) for some $q,k$ (see also~\cite{FineHarropJSL1957} for UIG graphs and \cite{GolumbicHammerJA1988} for UCA graphs).  

In~\cite{CostaDantasSankoffXuJBCS2012}, Costa et al.\ propose a specialized $O(n^2)$ time and space algorithm whose purpose is to find the minimum values $k$ and $q(k)$ such that a UIG graph $G$ is an induced subgraph of $P_{q(k)}^k$.  The reason for writing $q(k)$ dependent on $k$ is to be as truthful to~\cite{CostaDantasSankoffXuJBCS2012} as we can; they always write the number of vertices as a function on the power.  This is not important, though, as we know that $q$ in independent of $k$ by Pirlot's minimality Theorem~\cite{PirlotTaD1990}.  That is, $q$ is the minimum such that $G$ is an induced subgraph of $P_q^k$ for every possible $k$.  Mitas' algorithm could have been applied to obtain $k$ and $q$ in $O(n)$ time and space, under the assumption that it is correct.  Interestingly, Pirlot's Theorem and Mitas' algorithm predate~\cite{CostaDantasSankoffXuJBCS2012} for  at least fifteen years.  Moreover, \cite[Section 9]{Soulignac2010}, which is referenced within~\cite{CostaDantasSankoffXuJBCS2012}, mentions that Mitas' algorithm could be adapted to work when the input is a PIG model. The purpose of this section is to apply the minimization algorithms so as to find powers of paths and cycles supergraphs.

Let $\C_q^k$ (resp.\ $\MP_q^k$) be the $(2q, 2k+1)$-CA (resp.\ $(2k+1)$-IG) model that has an arc with beginning point $2i$ for every $0 \leq i < q$.  It is not hard to see that $\C_q^k$ (resp. $\MP_q^k$) is a $(1,0)$- and $(1,1)$-minimal model representing $C_q^k$ (resp.\ $P_q^k$).  We say that a $(\Circ, \Len)$-CA (resp.\ $\Len$-IG) model $\M^*$ is \emph{completable} when $\M^*$ can be obtained by removing arcs from $\C_q^k$ (resp.\ $\MP_q^k$) for some $k, q \geq 0$.  In such case, $\C_q^k$ (resp. $\MP_q^k$) is referred to as the \emph{completion} of $\M^*$, while $\M^*$ is said to be a \emph{$(k,q)$-extension} of $\M$ for every UCA (resp.\ UIG) model $\M$ equivalent to $\M^*$.  Note that $\M^*$ is completable if and only if:
\begin{enumerate}[({ext}$_1$)]
  \item $\Len$ is odd, \label{def:ext-1}
  \item $\Circ$ is even, and \label{def:ext-2}
  \item all its beginning points are even (thus $\M^*$ is a $(\Circ, \Len, 1,1)$-CA model).\label{def:ext-3}
\end{enumerate}
\newcommand{\extref}[1]{\itemref{ext}{ext}{#1}}

Under this new terminology, the result by Lin et al.~\cite{LinRautenbachSoulignacSzwarcfiterDAM2011} states that every UCA (resp.\ UIG) model $\M$ admits a $(k,q)$-extension $\M^*$ for some $k, q \geq 0$.  In analogy to minimal models, we say that $\M^*$ is a \emph{minimal extension} of $\M$ when $q \leq q'$ and $k \leq k'$ for every $(k',q')$-extension of $\M$.  The \emph{minimal power of a cycle (resp.\ path)} \MinPC (resp.\ \MinPP) problem consists of finding $\M^*$ when the UCA (resp.\ UIG) model $\M$ is given as input.  A priori, $\M$ could have no minimal extensions.  But, if $\M^*$ is the minimal extension of $\M$, then, clearly, $k$ and $q$ are the minimum values such that $G(\M)$ is an induced subgraph of $C_q^k$ (resp.\ $P_q^k$).  We now discuss how to solve \MinPC and \MinPP.

The fact that $\M$ admits a minimal extension follows by Lemma~\ref{lem:minimal model} and Theorem~\ref{thm:separation constraints}.  Indeed, if $\Len^*$ is the minimum odd number such that $\M$ is equivalent to a $(\Circ, \Len^*, 2)$-CA model, and $\Circ^*$ is the minimum even number such that $\M$ is equivalent to a $(\Circ^*, \Len, 1,1)$-CA model, then, by Lemma~\ref{lem:minimal model}, $\M$ is equivalent to a $(\Circ^*, \Len^*, 1,1)$-CA model.  Furthermore, $\SEP_{\Circ^*, \Len^*, 1,1}(A_i \to A_j)$ is even for every edge $A_i \to A_j$ of $\Syn$, by \sepref{1}--\sepref{4}.  Thus, all the beginning points of $\M^* = \U(\M, \Circ^*, \Len^*,  1,1)$ are even.  Then, $\M^*$ is completable by \extref{1}--\extref{3}, while it is equivalent to $\M$ by Theorem~\ref{thm:separation constraints}.  That is, $\M^*$ is the minimal extension of $\M$ and, thus, the solution to \MinPC.  The values $\Len^*$ and $\Circ^*$ can be found $O(n^4\log n)$ time with an algorithm similar to the one in Section~\ref{sec:minimal uca}.  

For the special case in which $\M$ is a UIG model, we observe that any $(1,1)$-minimal model $\M^*$ equivalent to $\M$ is a minimal extension of $\M$.  Just recall that the length $\Len^*$ of the arcs in $\M^*$ is equal to $\CONST(\W,1,1)-1$ for some path $\W$ of $\Syn(\M^*)$.  Since $\CONST(\W,1,1)$ is even, it follows that $\Len^*$ is odd and, thus, $\SEP_{\infty, \Len^*, 1,1}(A_i \to A_j)$ is even for every edge $A_i \to A_j$ of $\Syn(\M^*)$.  By \extref{1}--\extref{3}, this implies that $\M^*$ is an extension of $\M$ which, of course, is minimal by \minuigref{1}~and~\minuigref{2}.  By Theorem~\ref{thm:minimal algorithm}, \MinPP is solvable in $O(n^2)$ time and linear space.

\section{Further remarks}
\label{sec:remarks}

Synthetic graphs proved to be an important tool for studying how do the UIG representations of PIG graphs look like.  The generalization to PCA models is direct; the fact that some arcs wrap around the circle is not important for defining the synthetic graph. To represent the separation constraints that an equivalent UCA model must satisfy, all we had to include to Pirlot's original formulation was the variable $\Circ$ representing the circumference of the circle.  Generalizations of simple ideas from PIG to PCA graphs are not always as easy to obtain.  Unfortunately, Pirlot's ideas were introduced in the context of semiorders and were not exploited in the context of PCA graphs; the recognition problem of UCA graphs in polynomial time could have been solved more than a decade earlier.  In this closing section we provide some remarks and discuss some open problems.

Our definition of UCA descriptors states that every pair of beginning points should be separated by $\Thres + \BegDist$ distance.  An obvious generalization to \uRep and \BoundRepBoth is to replace $\BegDist$ with a function $\BegDist \colon \A \to \mathbb{Q}_{\geq 0}$ that indicates, for each arc $A_i$, the separation between $s(A_i)$ and the next beginning point $s(A_{i+1})$.  The reader can check that Theorem~\ref{thm:separation constraints} holds for this generalization as well.  All we need to do is to replace the value $\BegDist$ with $\BegDist(A_i)$ for each step $A_i \to A_{i+1}$.  Moreover, we can use similar functions to further separate $t(A_i)$ from $s(A_j)$ for every nose $A_i \to A_j$, and $s(A_i)$ from $t(A_j)$ for any hollow $A_i \to A_j$.  We did not consider these generalization for the sake of simplicity and notation.

In Section~\ref{sec:minimal uca} we gave a simple polynomial algorithm to transform a UCA model $\M$ into a minimal $(\Circ^*, \Len^*, \Thres, \BegDist)$-CA model.  The algorithm works by performing a linear search on $\Len^*$ and a binary search on $\Circ^*$.  An obvious idea to improve its running time is to replace the linear search on $\Len^*$ with a binary search.  Unfortunately, this idea is not feasible at first sight because we cannot claim 
\[L = \{\Len \in \mathbb{N} \mid \M \text{ is equivalent to a $(\Circ, \Len, \Thres, \BegDist)$-CA model for some $\Circ\in\mathbb{N}$}\}\]
to be a range.  For instance, $C^{4}_{11}$ admits a $(22, 9)$-CA model, but it admits no $(\Circ, 10)$-CA model, whatever value of $\Circ$ is.  This is just one more example of a property that is lost when the linear structure of PIG models is replaced by the circular structure of PCA graphs as $L = [\Len^*, \infty)$ when $\M$ is PIG.

As calculated in Section~\ref{sec:minimal uca}, the running time of the minimization algorithm is $O((\Thres+\BegDist)n^4\log (n + \Thres + \BegDist))$.  This bound is not tight, as the actual running time is $O(\Len^* n^2\log (n\Len^*))$, and $\Len^*$ could be much lower than $(\Thres+\BegDist)n^2$.  As a matter of fact, we developed a simple program for testing if a UCA model is equivalent to some $(\Circ, 2n)$-CA model.  We tested it on many input UCA models and, in all cases, the program was successful.

In Section~\ref{sec:minimal uig} we fixed Mitas' algorithm so as to solve the minimization problem for UIG models.  Unfortunately, the running time of the patched algorithm is $O(n^2)$.  There are two bottlenecks in this algorithm.  First, we have to compute the minimum length value $\Len^*$.  Then, we have to apply the Bellman-Ford algorithm to compute the actual model.  With respect to the space complexity, it is not hard to observe that $\Len^*$ can be computed in unambiguous logspace.  Indeed, all we have to do is to find the distance between the leftmost and rightmost arcs for every height.  As the canonical drawing is a plane graph with $O(1)$ vertices with $0$ in-degree, this problem requires unambiguous logspace~\cite{LimayeMahajanNimbhorkarCJTCS2010}.  Finding logspace algorithms to compute $\Len^*$ and the minimal model remain as open problems.

Finally, it should be noted that a UCA graph may admit many \emph{non-equivalent} minimal UCA models.  Indeed, a UCA graph may admit an exponential number of non-equivalent models, each of which is equivalent to a minimal UCA model.  It makes sense, then, to say that a model $\M$ is \emph{minimum} when it satisfies \minucaref{1} and \minucaref{2} for every model $\M'$ such that $G(\M)$ is isomorphic to $G(\M')$.  As it was noted by Huang~\cite{HuangJCTSB1995}, every connected and co-connected PCA graph admits a unique PCA model, up to equivalence and full reversal.  Thus, any minimal model of a connected and co-connected PCA graph is minimum.  Similarly, every disconnected PCA graph is PIG, and all its models can be obtained from a model $\M$ by exchanging the order in which their components appear from $0$, and reversing some of the components.  Thus, again, any minimal model of $G(\M)$ is minimum.  Co-disconnected PCA graphs share a similar property: all their PCA models can be obtained from a PCA model $\M$ by exchanging the order in which its co-components appear, plus reversing some co-components~\cite{HuangJCTSB1995}.  Thus, one is tempted to think that all the minimal PCA models are minimum, yet this is not the case.  Figure~\ref{fig:minimum model} shows two $1$-minimal $(18, 7)$-CA and $(20, 8)$-CA models that represent the graph whose co-components are $P_2 \cup P_1$ and $P_4$.  We leave as open the problem of computing the minimum UCA model.

\begin{figure}
  \hfil  \includegraphics{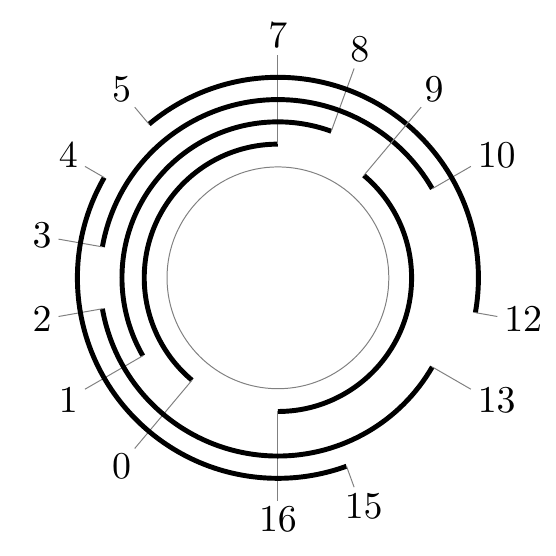} \hfil \includegraphics{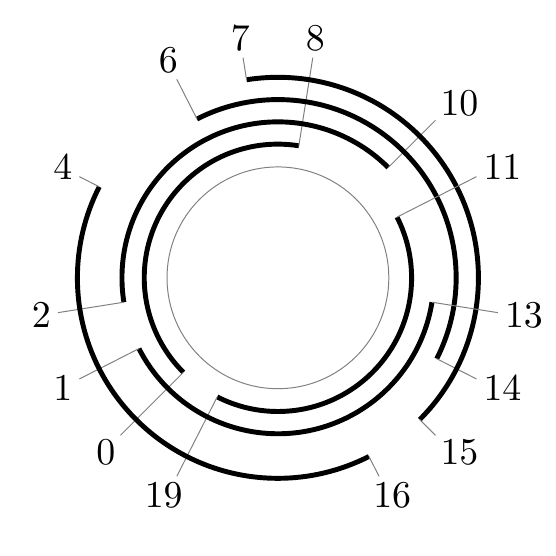} \hfil
  
  \caption{Two minimal UCA models representing the same graph.}\label{fig:minimum model}
\end{figure}

\small

\begin{thebibliography}{35}
\providecommand{\natexlab}[1]{#1}
\providecommand{\url}[1]{\texttt{#1}}
\expandafter\ifx\csname urlstyle\endcsname\relax
  \providecommand{\doi}[1]{doi: #1}\else
  \providecommand{\doi}{doi: \begingroup \urlstyle{rm}\Url}\fi

\bibitem[Balko et~al.(2013)Balko, Klavík, and Otachi]{BalkoKlavikOtachi2013}
M.~Balko, P.~Klavík, and Y.~Otachi.
\newblock Bounded representations of interval and proper interval graphs.
\newblock In L.~Cai, S.-W. Cheng, and T.-W. Lam, editors, \emph{Algorithms and
  Computation}, vol. 8283 of \emph{Lecture Notes in Computer Science}, pp.
  535--546. Springer, 2013.
\newblock \doi{10.1007/978-3-642-45030-3_50}.

\bibitem[Bl{\"{a}}sius and Rutter(2013)]{BlasiusRutter2013}
T.~Bl{\"{a}}sius and I.~Rutter.
\newblock Simultaneous pq-ordering with applications to constrained embedding
  problems.
\newblock In S.~Khanna, editor, \emph{Proceedings of the Twenty-Fourth Annual
  {ACM-SIAM} Symposium on Discrete Algorithms, {SODA} 2013, New Orleans,
  Louisiana, USA, January 6-8, 2013}, pp. 1030--1043. {SIAM}, 2013.
\newblock \doi{10.1137/1.9781611973105.74}.

\bibitem[Chaplick et~al.()Chaplick, Dorbec, Kratochvíl, Montassier, and
  Stacho]{ChaplickDorbecKratochvilMontassierStacho}
S.~Chaplick, P.~Dorbec, J.~Kratochvíl, M.~Montassier, and J.~Stacho.
\newblock Contact representations of planar graph: Rebuilding is hard.
\newblock In \emph{Graph-Theoretic Concepts in Computer Science}, Lecture Notes
  in Comput. Sci. Springer, Berlin.

\bibitem[Chaplick et~al.(2013)Chaplick, Fulek, and
  Klavík]{ChaplickFulekKlavik2013}
S.~Chaplick, R.~Fulek, and P.~Klavík.
\newblock Extending partial representations of circle graphs.
\newblock In S.~Wismath and A.~Wolff, editors, \emph{Graph Drawing}, vol. 8242
  of \emph{Lecture Notes in Computer Science}, pp. 131--142. Springer, 2013.
\newblock \doi{10.1007/978-3-319-03841-4_12}.

\bibitem[Cormen et~al.(2009)Cormen, Leiserson, Rivest, and
  Stein]{CormenLeisersonRivestStein2009}
T.~H. Cormen, C.~E. Leiserson, R.~L. Rivest, and C.~Stein.
\newblock \emph{Introduction to algorithms}.
\newblock MIT Press, Cambridge, MA, third edition, 2009.

\bibitem[Corneil et~al.(1995)Corneil, Kim, Natarajan, Olariu, and
  Sprague]{CorneilKimNatarajanOlariuSpragueIPL1995}
D.~G. Corneil, H.~Kim, S.~Natarajan, S.~Olariu, and A.~P. Sprague.
\newblock Simple linear time recognition of unit interval graphs.
\newblock \emph{Inform. Process. Lett.}, 55\penalty0 (2):\penalty0 99--104,
  1995.
\newblock \doi{10.1016/0020-0190(95)00046-F}.

\bibitem[Costa et~al.(2012)Costa, Dantas, Sankoff, and
  Xu]{CostaDantasSankoffXuJBCS2012}
V.~Costa, S.~Dantas, D.~Sankoff, and X.~Xu.
\newblock Gene clusters as intersections of powers of paths.
\newblock \emph{J. Braz. Comput. Soc.}, 18\penalty0 (2):\penalty0 129--136,
  2012.
\newblock \doi{10.1007/s13173-012-0064-8}.

\bibitem[Dur{\'a}n et~al.(2006)Dur{\'a}n, Gravano, McConnell, Spinrad, and
  Tucker]{DuranGravanoMcConnellSpinradTuckerJA2006}
G.~Dur{\'a}n, A.~Gravano, R.~M. McConnell, J.~Spinrad, and A.~Tucker.
\newblock Polynomial time recognition of unit circular-arc graphs.
\newblock \emph{J. Algorithms}, 58\penalty0 (1):\penalty0 67--78, 2006.
\newblock \doi{10.1016/j.jalgor.2004.08.003}.

\bibitem[Fine and Harrop(1957)]{FineHarropJSL1957}
N.~J. Fine and R.~Harrop.
\newblock Uniformization of linear arrays.
\newblock \emph{J. Symb. Logic}, 22:\penalty0 130--140, 1957.

\bibitem[Gardi(2007)]{GardiDM2007}
F.~Gardi.
\newblock The {R}oberts characterization of proper and unit interval graphs.
\newblock \emph{Discrete Math.}, 307\penalty0 (22):\penalty0 2906--2908, 2007.
\newblock \doi{10.1016/j.disc.2006.04.043}.

\bibitem[Golumbic and Hammer(1988)]{GolumbicHammerJA1988}
M.~C. Golumbic and P.~L. Hammer.
\newblock Stability in circular arc graphs.
\newblock \emph{J. Algorithms}, 9\penalty0 (3):\penalty0 314--320, 1988.
\newblock \doi{10.1016/0196-6774(88)90023-5}.

\bibitem[Goodman(1977)]{Goodman1977}
N.~Goodman.
\newblock \emph{The Structure of Appearance}.
\newblock Boston Studies in the Philosophy and History of Science. Springer
  Netherlands, 3rd edition, 1977.

\bibitem[Hell and Huang(2004/05)]{HellHuangSJDM2004/05}
P.~Hell and J.~Huang.
\newblock Certifying {L}ex{BFS} recognition algorithms for proper interval
  graphs and proper interval bigraphs.
\newblock \emph{SIAM J. Discrete Math.}, 18\penalty0 (3):\penalty0 554--570
  (electronic), 2004/05.
\newblock \doi{10.1137/S0895480103430259}.

\bibitem[Huang(1995)]{HuangJCTSB1995}
J.~Huang.
\newblock On the structure of local tournaments.
\newblock \emph{J. Combin. Theory Ser. B}, 63\penalty0 (2):\penalty0 200--221,
  1995.
\newblock \doi{10.1006/jctb.1995.1016}.

\bibitem[Kaplan and Nussbaum(2009)]{KaplanNussbaumDAM2009}
H.~Kaplan and Y.~Nussbaum.
\newblock Certifying algorithms for recognizing proper circular-arc graphs and
  unit circular-arc graphs.
\newblock \emph{Discrete Appl. Math.}, 157\penalty0 (15):\penalty0 3216--3230,
  2009.
\newblock \doi{10.1016/j.dam.2009.07.002}.

\bibitem[Klav{\'{\i}}k et~al.(2011)Klav{\'{\i}}k, Kratochv{\'{\i}}l, and
  Vysko{\v{c}}il]{KlavikKratochvilVyskovcil2011}
P.~Klav{\'{\i}}k, J.~Kratochv{\'{\i}}l, and T.~Vysko{\v{c}}il.
\newblock Extending partial representations of interval graphs.
\newblock In \emph{Theory and applications of models of computation}, vol. 6648
  of \emph{Lecture Notes in Comput. Sci.}, pp. 276--285. Springer, Heidelberg,
  2011.
\newblock \doi{10.1007/978-3-642-20877-5_28}.

\bibitem[Klav{\'{\i}}k et~al.(2012)Klav{\'{\i}}k, Kratochv{\'{\i}}l, Krawczyk,
  and Walczak]{KlavikKratochvilKrawczykWalczak2012}
P.~Klav{\'{\i}}k, J.~Kratochv{\'{\i}}l, T.~Krawczyk, and B.~Walczak.
\newblock Extending partial representations of function graphs and permutation
  graphs.
\newblock In \emph{Algorithms---{ESA} 2012}, vol. 7501 of \emph{Lecture Notes
  in Comput. Sci.}, pp. 671--682. Springer, Heidelberg, 2012.
\newblock \doi{10.1007/978-3-642-33090-2_58}.

\bibitem[Klav{\'{\i}}k et~al.(2014)Klav{\'{\i}}k, Kratochv{\'{\i}}l, Otachi,
  Rutter, Saitoh, Saumell, and
  Vyskocil]{KlavikKratochvilOtachiRutterSaitohSaumellVyskocilC2014}
P.~Klav{\'{\i}}k, J.~Kratochv{\'{\i}}l, Y.~Otachi, I.~Rutter, T.~Saitoh,
  M.~Saumell, and T.~Vyskocil.
\newblock Extending partial representations of proper and unit interval graphs.
\newblock \emph{CoRR}, abs/1207.6960v2, 2014.

\bibitem[Klavík et~al.(2014)Klavík, Kratochvíl, Otachi, Rutter, Saitoh,
  Saumell, and
  Vyskočil]{KlavikKratochvilOtachiRutterSaitohSaumellVyskocil2014}
P.~Klavík, J.~Kratochvíl, Y.~Otachi, I.~Rutter, T.~Saitoh, M.~Saumell, and
  T.~Vyskočil.
\newblock Extending partial representations of proper and unit interval graphs.
\newblock In R.~Ravi and I.~Gørtz, editors, \emph{Algorithm Theory – SWAT
  2014}, vol. 8503 of \emph{Lecture Notes in Computer Science}, pp. 253--264.
  Springer, 2014.
\newblock \doi{10.1007/978-3-319-08404-6_22}.

\bibitem[K{\"o}bler et~al.(2011)K{\"o}bler, Kuhnert, Laubner, and
  Verbitsky]{KoblerKuhnertLaubnerVerbitskySJC2011}
J.~K{\"o}bler, S.~Kuhnert, B.~Laubner, and O.~Verbitsky.
\newblock Interval graphs: canonical representations in logspace.
\newblock \emph{SIAM J. Comput.}, 40\penalty0 (5):\penalty0 1292--1315, 2011.
\newblock \doi{10.1137/10080395X}.

\bibitem[K{\"o}bler et~al.(2012)K{\"o}bler, Kuhnert, and
  Verbitsky]{KoblerKuhnertVerbitsky2012}
J.~K{\"o}bler, S.~Kuhnert, and O.~Verbitsky.
\newblock Solving the canonical representation and star system problems for
  proper circular-arc graphs in logspace.
\newblock In \emph{32nd {I}nternational {C}onference on {F}oundations of
  {S}oftware {T}echnology and {T}heoretical {C}omputer {S}cience}, vol.~18 of
  \emph{LIPIcs. Leibniz Int. Proc. Inform.}, pp. 387--399. Schloss Dagstuhl.
  Leibniz-Zent. Inform., Wadern, 2012.

\bibitem[K{\"o}bler et~al.(2013)K{\"o}bler, Kuhnert, and
  Verbitsky]{KoblerKuhnertVerbitskyC2013}
J.~K{\"o}bler, S.~Kuhnert, and O.~Verbitsky.
\newblock Solving the canonical representation and star system problems for
  proper circular-arc graphs in log-space.
\newblock \emph{CoRR}, abs/1202.4406v5, 2013.

\bibitem[Limaye et~al.(2010)Limaye, Mahajan, and
  Nimbhorkar]{LimayeMahajanNimbhorkarCJTCS2010}
N.~Limaye, M.~Mahajan, and P.~Nimbhorkar.
\newblock Longest paths in planar {DAG}s in unambiguous log-space.
\newblock \emph{Chic. J. Theoret. Comput. Sci.}, pp. Article 8, 15, 2010.

\bibitem[Lin and Szwarcfiter(2008)]{LinSzwarcfiterSJDM2008}
M.~C. Lin and J.~L. Szwarcfiter.
\newblock Unit circular-arc graph representations and feasible circulations.
\newblock \emph{SIAM J. Discrete Math.}, 22\penalty0 (1):\penalty0 409--423,
  2008.
\newblock \doi{10.1137/060650805}.

\bibitem[Lin et~al.(2009)Lin, Soulignac, and
  Szwarcfiter]{LinSoulignacSzwarcfiter2009}
M.~C. Lin, F.~J. Soulignac, and J.~L. Szwarcfiter.
\newblock Short models for unit interval graphs.
\newblock In \emph{L{AGOS}'09---{V} {L}atin-{A}merican {A}lgorithms, {G}raphs
  and {O}ptimization {S}ymposium}, vol.~35 of \emph{Electron. Notes Discrete
  Math.}, pp. 247--255. Elsevier Sci. B. V., Amsterdam, 2009.
\newblock \doi{10.1016/j.endm.2009.11.041}.

\bibitem[Lin et~al.(2011)Lin, Rautenbach, Soulignac, and
  Szwarcfiter]{LinRautenbachSoulignacSzwarcfiterDAM2011}
M.~C. Lin, D.~Rautenbach, F.~J. Soulignac, and J.~L. Szwarcfiter.
\newblock Powers of cycles, powers of paths, and distance graphs.
\newblock \emph{Discrete Appl. Math.}, 159\penalty0 (7):\penalty0 621--627,
  2011.
\newblock \doi{10.1016/j.dam.2010.03.012}.

\bibitem[Mitas(1994)]{Mitas1994}
J.~Mitas.
\newblock Minimal representation of semiorders with intervals of same length.
\newblock In \emph{Orders, algorithms, and applications ({L}yon, 1994)}, vol.
  831 of \emph{Lecture Notes in Comput. Sci.}, pp. 162--175. Springer, Berlin,
  1994.
\newblock \doi{10.1007/BFb0019433}.

\bibitem[Pirlot(1990)]{PirlotTaD1990}
M.~Pirlot.
\newblock Minimal representation of a semiorder.
\newblock \emph{Theory and Decision}, 28\penalty0 (2):\penalty0 109--141, 1990.
\newblock \doi{10.1007/BF00160932}.

\bibitem[Pirlot(1991)]{PirlotDAM1991}
M.~Pirlot.
\newblock Synthetic description of a semiorder.
\newblock \emph{Discrete Appl. Math.}, 31\penalty0 (3):\penalty0 299--308,
  1991.
\newblock \doi{10.1016/0166-218X(91)90057-4}.

\bibitem[Pirlot and Vincke(1997)]{PirlotVincke1997}
M.~Pirlot and P.~Vincke.
\newblock \emph{Semiorders}, vol.~36 of \emph{Theory and Decision Library.
  Series B: Mathematical and Statistical Methods}.
\newblock Kluwer Academic Publishers Group, Dordrecht, 1997.
\newblock \doi{10.1007/978-94-015-8883-6}.

\bibitem[Roberts(1969)]{Roberts1969}
F.~S. Roberts.
\newblock Indifference graphs.
\newblock In \emph{Proof {T}echniques in {G}raph {T}heory ({P}roc. {S}econd
  {A}nn {A}rbor {G}raph {T}heory {C}onf., {A}nn {A}rbor, {M}ich., 1968)}, pp.
  139--146. Academic Press, New York, 1969.

\bibitem[Soulignac(2010)]{Soulignac2010}
F.~J. Soulignac.
\newblock \emph{On proper and {H}elly circular-arc graphs}.
\newblock PhD thesis, Universidad de Buenos Aires, March 2010.

\bibitem[Soulignac(2013)]{SoulignacA2013}
F.~J. Soulignac.
\newblock Fully dynamic recognition of proper circular-arc graphs.
\newblock \emph{Algorithmica}, 2013.
\newblock \doi{10.1007/s00453-013-9835-7}.

\bibitem[Stolee and Vinodchandran(2012)]{StoleeVinodchandran2012}
D.~Stolee and N.~V. Vinodchandran.
\newblock Space-efficient algorithms for reachability in surface-embedded
  graphs.
\newblock In \emph{2012 {IEEE} 27th {C}onference on {C}omputational
  {C}omplexity---{CCC} 2012}, pp. 326--333. IEEE Computer Soc., Los Alamitos,
  CA, 2012.
\newblock \doi{10.1109/CCC.2012.15}.

\bibitem[Tucker(1974)]{TuckerDM1974}
A.~Tucker.
\newblock Structure theorems for some circular-arc graphs.
\newblock \emph{Discrete Math.}, 7:\penalty0 167--195, 1974.

\end{thebibliography}

\end{document}